\documentclass[11pt]{article}

\usepackage{epsf}
\usepackage{verbatim}
\usepackage{makeidx}
\usepackage{graphicx}
\usepackage{amssymb}
\usepackage{amsmath}
\usepackage{latexsym}
\usepackage{amsbsy}
\usepackage{epsfig}

\usepackage{theorem}

\usepackage{epsf}
\usepackage{verbatim}
\usepackage{makeidx}
\usepackage{graphicx}
\usepackage{amssymb}
\usepackage{amsmath}
\usepackage{latexsym}
\usepackage{amsbsy}
\usepackage{epsfig}

\newtheorem{theorem}{Theorem}

\newtheorem{proposition}{Proposition}

\newtheorem{definition}{Definition}

\newcommand{\commentout}[1]{}
\newcommand{\prs}{\par}

\newenvironment{proof}{\noindent{\bf Proof}:}{\mbox{$\Box$}
\prs\vspace{3mm}\prs}

\newcounter{myenumeratecounter}

\newtheorem{examplehidden}{Example}
\newenvironment{example}{
\begin{examplehidden}
\em
}
{\end{examplehidden}}

\newcommand{\ave}[1]{\ensuremath{\langle #1\rangle}}

\newcommand{\determines}{\ensuremath{\rightsquigarrow}}

\newcommand{\cA}{\ensuremath{\mathcal A}}

\newcommand{\cC}{\ensuremath{\mathcal C}}

\newcommand{\cE}{\ensuremath{\mathcal E}}

\newcommand{\cM}{\ensuremath{\mathcal M}}

\newcommand{\cP}{\ensuremath{\mathcal P}}
\newcommand{\cQ}{\ensuremath{\mathcal Q}}

\newcommand{\cS}{\ensuremath{\mathcal S}}
\newcommand{\cT}{\ensuremath{\mathcal T}}
\newcommand{\cU}{\ensuremath{\mathcal U}}
\newcommand{\cV}{\ensuremath{\mathcal V}}

\newcommand{\cX}{\ensuremath{\mathcal X}}
\newcommand{\cY}{\ensuremath{\mathcal Y}}
\newcommand{\cZ}{\ensuremath{\mathcal Z}}

\newcommand{\indicator}{\ensuremath{{\mathbf 1}}}

\newcommand{\reals}{{\mathbb R}}

\newcommand{\naturals}{{\mathbb N}}

\newcommand{\Exp}{\ensuremath{\text{\rm E}}}

\newcommand{\range}{\ensuremath \mathbf U}

\usepackage{color,stmaryrd}
\usepackage[round,comma]{natbib}
\makeatletter
\begin{document}


\newlength{\ownboxwidth} \newlength{\ownleftboxwidth}
  \newlength{\ownrightboxwidth} \setlength{\ownboxwidth}{\textwidth}
  \addtolength{\ownboxwidth}{-0.1 cm} \setlength{\ownleftboxwidth}{0.1
    cm} \setlength{\ownrightboxwidth}{0.01 cm}
 
\def\pperp{\perp\!\!\!\perp}

\newcommand{\wksafe}[1]{\ensuremath{ \llbracket #1 \rrbracket}}
\newcommand{\prag}[1]{\ensuremath{\tilde{#1}}}
\newcommand{\pragset}[1]{\ensuremath{\tilde{#1}}}
\newcommand{\pur}[1]{\ensuremath{\tilde{#1}}}
\newcommand{\tpmf}[1]{\ensuremath{\tilde{p}_{\tiny [#1]}}}
\newcommand{\pmf}[1]{\ensuremath{p_{\tiny [#1]}}}
\newcommand{\cali}[1]{[#1]}
\newcommand{\cdf}{\prag{F}_{\small [U|V]}}
\newcommand{\cs}{\cC}
\newcommand{\comp}{\text{\sc comp}}
\newcommand{\ct}{{\cC}^{[a,b]}}
\renewcommand{\range}[1]{\text{\sc range}(#1)}
\newcommand{\cod}[1]{{\mathcal{#1}}}
\newcommand{\todo}[1]{{\color {blue} #1 \ }}
\newcommand{\Pdetermines}[1]{\ensuremath{\rightsquigarrow}_{#1}}
\newcommand{\fPdetermines}[2]{\ensuremath{\overset{\small #1}{\rightsquigarrow}}_{#2}}
\newcommand{\fdetermines}{\overset{\small f}{\determines}} 
\newcommand{\onetoones}{\ensuremath{\leftrightsquigarrow}}
\newcommand{\Ponetoones}[1]{\ensuremath{\leftrightsquigarrow}_{#1}}
\newcommand{\cPonetoones}{\ensuremath{\leftrightsquigarrow}_{\cP^*}}

\newcommand{\Pmarg}{\prag{P}_{\text{\sc marg}}}
\newcommand{\we}{}
\newcommand{\cont}{\text{\sc cont}}
\newcommand{\Pstand}{\prag{P}_{\text{\sc standard}}}
\newcommand{\ind}[1]{\text{\sc ind}{(#1)}}
\newcommand{\oss}{\text{\bf  oss}}
\newcommand{\tss}{\text{\bf  tss}}
\newcommand{\allx}{\bigcup_{s \in \cS} \cX_s}
\newcommand{\ally}{\bigcup_{t \in\cT} \cY_t}
\newcommand{\longversion}[1]{}
\newcommand{\longversionnew}[1]{}
\newcommand{\preds}{\prag{P}_{[Y_t|X_s]}}
\newcommand{\predm}{\prag{P}_{[Y_t|{\bf 0}]}}
\newcommand{\support}{\text{\sc supp}}
\newcommand{\fun}[1]{\bar{#1}}
\newcommand{\noproc}[1]{}
\newcommand{\proconly}[1]{#1}
\newcommand{\shortp}[1]{#1}
\newcommand{\heads}{\text{\sc h}}
\newcommand{\tails}{\text{\sc t}}
\newcommand{\direct}{\text{\sc d}}
\newcommand{\indirect}{\text{\sc i}}
\newcommand{\lzo}{L_{01}}
\newcommand{\ps}{\text{\sc \bf PS}}

\title{Safe Probability}


%
\author{Peter Gr{\"u}nwald}
%
%

\maketitle              

\begin{abstract}
  We formalize the idea of probability distributions that lead to
  reliable predictions about some, but not all aspects of a domain.
  The resulting notion of `safety' provides a fresh perspective on
  foundational issues in statistics, providing a middle ground between
  imprecise probability and multiple-prior models on the one hand and
  strictly Bayesian approaches on the other. It also allows us to
  formalize fiducial distributions in terms of the set of random
  variables that they can safely predict, thus taking some of the
  sting out of the fiducial idea. By restricting probabilistic
  inference to safe uses, one also automatically avoids paradoxes such
  as the Monty Hall problem. Safety comes in a variety of degrees, such as   
  `validity' (the strongest notion),  `calibration', `confidence safety' and `unbiasedness' (almost the weakest notion). 
\end{abstract}
\section{Introduction}
We formalize the idea of probability distributions that lead to
reliable predictions about some, but not all aspects of a domain. Very
broadly speaking, we call a distribution $\tilde{P}$ {\em safe\/} for
predicting random variable $U$ given random variable $V$ if
predictions concerning $U$ based on $\tilde{P}(U|V)$ tend to be as
good as one would expect them to be if $\tilde{P}$ were an accurate
description of one's uncertainty, even if $\tilde{P}$ may not
represent one's actual beliefs, let alone the truth.
Our formalization of this notion of `safety' 
has repercussions for the foundations of
statistics, 
providing a joint perspective on issues hitherto viewed as distinct:
%

\paragraph{1. All models are wrong...\protect\footnote{...yet some are
    useful, as famously remarked by \cite{Box79}.}}  Some statistical
models are evidently both entirely wrong yet very useful. For example,
in some highly successful applications of Bayesian statistics, such as
latent Dirichlet allocation for topic modeling \citep{blei2003latent},
one assumes that natural language text is i.i.d., which is fine for
the task at hand (topic modeling) --- yet no-one would want to use
these models for predicting the next word of a text given the
past. Yet, one can use a Bayesian posterior to make such predictions
any way --- Bayesian inference has no mechanism to distinguish between
`safe' and `unsafe' inferences. Safe probability allows us to impose
such a distinction.

\paragraph{2. The Eternal Discussion\protect\footnote{When the single-vs. multiple-prior issue came up in a discussion 
on the {\em decision-theory forum\/} mailing list, the well-known economist I. Gilboa referred to it as `the eternal discussion'.}}
More generally, representing uncertainty by a single distribution, as
is standard in Bayesian inference, implies a willingness to make
definite predictions about random variables that, some claim, one
really knows nothing about. Disagreement on this issue goes back at
least to \cite{Keynes21} and \cite{Ramsey31}, has led many economists
to sympathize with {\em multiple-prior models\/} \citep{GilboaS89} and
some statisticians to embrace the related {\em imprecise
  probability\/} \citep{Walley91,augustin2014introduction} in which so-called `Knightian'
uncertainty is modeled by a {\em set\/} $\cP^*$ of distributions. But
imprecise probability is not without problems of its own, an
important one being {\em dilation\/} (Example~\ref{ex:dilation}
below). Safe probability can be understood as starting from a set
$\cP^*$, but then {\em mapping\/} the set of distributions to a single
distribution, where the mapping invoked may depend on the prediction task
at hand --- thus avoiding both dilation and overly precise
predictions. The use of such mappings has been advocated before, under
the name {\em pignistic transformation\/}
\citep{Smets89,hampel2001outline}, but a general theory for
constructing and evaluating them has been lacking (see also
Section~\ref{sec:conclusion}).

\paragraph{3. Fisher's Biggest Blunder\protect\footnote{While Fisher
    is generally regarded as (one of) the greatest statisticians of
    all time, fiducial inference is often considered to be his `big
    blunder' --- see \cite{Hampel06} and \cite{Efron96}, who writes
    {\em Maybe Fisher's biggest blunder will become a big hit in the
      21st century!}}}  Fisher (1930) \nocite{Fisher30} introduced
{\em fiducial inference}, a method to come up with a `posterior'
$\prag{P}(\theta \mid X^n)$ on a model's parameter space based on data
$X^n$, but without anything like a `prior', in an approach to
statistics that was neither Bayesian nor frequentist. The approach
turned out problematic however, and, despite progress on related {\em 
structural inference\/}  \citep{Fraser68,Fraser79} was largely abandoned. 
Recently, however, fiducial distributions have made a comeback
\citep{hannig2009generalized,taraldsen2013fiducial,martin2013inferential,veronese2015fiducial},
in some instances with a more modest, frequentist interpretation as
{\em confidence distributions\/} \citep{SchwederH02,SchwederH16}. 
As noted by
\cite{xie2013confidence}, these `contain a wealth of information for
inference', e.g. to determine valid confidence intervals and unbiased
estimation of the median, but their interpretation remains difficult,
viz. the insistence by \cite{Hampel06,xie2013confidence} and many
others that, although $\prag{P}(\cdot \mid X^n)$ is defined as a
distribution on the parameter space, the parameter itself is not
random. Safe probability offers an alternative perspective, where the
insistence that `$\theta$ is not random' is replaced by the weaker
(and perhaps liberating) statement that `we can treat $\theta$ as
random' {\em as long as we restrict ourselves to safe inferences about
  it}' --- in Section~\ref{sec:confidence} we determine precisely what these
safe inferences are  and how they fit into a general hierarchy:

\paragraph{4. The Hierarchy} Pursuing the idea that some distributions
are reliable for a smaller subset of random variables/prediction tasks
than others, leads to a natural {\em hierarchy\/} of safeties --- a
first taste of which is in Figure~\ref{fig:figure} on
page~\pageref{fig:figure}, with  notations explained later.
At the top are distributions that are fully reliable for whatever task
one has in mind; at the bottom those that are reliable only for a
single task in a weak, average sense. In between there is a natural
place for distributions that are {\em calibrated\/}
(Example~\ref{ex:calibration} below), that are {\em confidence--safe\/}
(i.e. valid confidence distributions) and that are {\em optimal for
  squared-error prediction}. 

\paragraph{5.  ``The concept of a conditional probability with regard to an isolated hypothesis...\protect\footnote{...
whose probability equals 0 is inadmissible,'' as remarked by \cite{Kolmogorov33}. As will be seen, safe probability suggests an even more radical statement related to the Monty Hall sanity check. } } 
Upon first hearing of the Monty Hall (quiz master, three doors)
problem \citep{vossavant90,Gill11}, most people naively think that the
probability of winning the car is the same whether one switches doors
or not. Most can eventually, after much arguing, be convinced that
this is wrong, but wouldn't it be nice to have a simple sanity check
that {\em immediately\/} tells you that the naive answer must be
wrong, without even pondering the `right' way to approach the problem?
Safe probability provides such a check: one can immediately tell that
the naive answer is {\em not safe}, and thus cannot be right. Such a
check is applicable more generally, whenever conditioning on events
rather than on random variables (Example~\ref{ex:montyhall} and
Section~\ref{sec:general}).

\paragraph{6. ``Could Neyman, Jeffreys and Fisher have agreed on testing?\protect\footnote{...'', as asked by Jim \cite{Berger03}. }}
\cite{ryabko2005using} shows that sequences of 0s and 1s produced by
standard random number generators can be substantially compressed by
standard data compression algorithms such as {\tt rar} or {\tt zip}.  While this
is clear evidence that such sequences are not random, this method is
neither a valid Neyman-Pearson hypothesis test nor a valid Bayesian
test (in the tradition of Jeffreys). The reason is that both these
standard paradigms require the existence of an {\em alternative
  statistical model}, and start out by the assumption that, if the
null model (i.i.d. Bernoulli (1/2)) is incorrect, then the alternative
must be correct. However, there is no clear sense in which {\tt zip}
could be `correct' --- see Section~\ref{sec:conclusion}. There is a
third testing paradigm, due to Fisher, which does view testing as
accumulating evidence against $h_0$, and not necessarily as confirming
some precisely specified $h_1$. Yet Fisher's paradigm is not without
serious problems either --- see Section~\ref{sec:conclusion}.

\cite{BergerBW94} started a line of work culminating in
\cite{Berger03}, who presents tests that have interpretations in all
three paradigms and that avoid some of the problems of their original
implementations. However, it is essentially an objective Bayes
approach and thus inevitably, strong evidence against $h_0$ implies a
high posterior probability that $h_1$ is true. If one is really doing
Fisherian testing, this is unwanted.  Using the idea of safety, we can
extend Berger's paradigm by stipulating the inferences for which we
think it is safe: roughly speaking, if we are in a Fisherian
set-up, then we declare all inferences conditional on $h_1$ to be
unsafe, and inferences conditional on $h_0$ to be safe; if we really
believe that $h_1$ may represent the state of the world, we can
declare inferences conditional on $h_1$ to be safe.  But much more is
possible using safe probability --- a DM can decide, on a case by case
basis, what inferences based on her tests would be safe, and under
what situations the test results itself are safe --- for example, some
tests remain safe under optional stopping, whereas others (even
Bayesian ones!) do not. While we will report on this application of
safety (which comprises a long paper in itself) elsewhere,
we will briefly return to it in the conclusion.

\paragraph{7. Further Applications: Objective Bayes, Epistemic Probability}
Apart from the applications above, the results in this paper suggest
that safe probability be used to formalize the status of default
priors in {\em objective Bayesian\/} inferences, and to enable an
alternative look at {\em epistemic probability}. But this remains a
topic for future work, to which we briefly return at the end of the
paper.

\paragraph{The Dream}
Imagine a world in which one would require any statistical analysis
--- whether it be testing, prediction, regression, density estimation
or anything else --- to be accompanied by a {\em safety
  statement}. Such a statement should list what inferences, the
analysists think, can be safely made based on the conclusion of the
analysis, and in what formal `safety' sense. Is the alternative $h_1$
really true even though $h_0$ is found to be false? Is the suggested
predictive distribution valid or merely calibrated? Is the posterior
really just good for making predictions via the predictive
distribution, or is it confidence-safe, or is it generally safe? Does
the inferred regression function only work well on covariates drawn
randomly from the same distribution, or also under covariate shift?
(an application of safety we did not address here but which we can
easily incorporate). The present, initial formulation of safe
probability is too complicated to have any realistic hopes for a
practice like this to emerge, but I can't help hoping that the ideas
can be simplified substantially, and a safer practice of statistics
might emerge.

\ \\
Starting with \cite{Grunwald99a}, my own work --- often in
collaboration with J. Halpern --- has regularly used the idea of
`safety', for example in the context of Maximum Entropy inference
\citep{Grunwald00a}, and also dilation \citep{GrunwaldH04},
calibration \citep{GrunwaldH11}, and probability puzzles like Monty
Hall \citep{GrunwaldH03,Grunwald13}. However, the insights of earlier
papers were very partial and scattered, and the present paper presents
for the first time a general formalism, definitions and a
hierarchy. It is also the first one to make a connection to confidence
distributions and pivots.
\subsection{Informal Overview}
Below we explain the basic ideas using three recurring
examples. We assume that we are given a set of distributions $\cP^*$
on some space of outcomes $\cZ$.  Under a frequentist interpretation,
$\cP^*$ is the set of distributions that we regard as `potentially
true'; under a subjectivist interpretation, it is the {\em credal
  set\/} that describes our uncertainty or `beliefs'; all developments
below work under both interpretations.

All probability distributions mentioned below are either an element of
$\cP^*$, or they are a {\em pragmatic distribution\/} $\prag{P}$,
which some decision-maker (DM) uses to predict the outcomes of some
variable $U$ given the value of some other variable $V$,where both $U$
and $V$ are random quantities defined on $\cZ$.  $\prag{P}$ is also
used to estimate the quality of such predictions. $\prag{P}$ (which
may be, but is not always in $\cP^*$) is `pragmatic' because we assume
from the outset that some element of $\cP^*$ might actually lead to
better predictions --- we just do not know which one.
\begin{figure}[t]{\hspace*{0.25\textwidth}
\includegraphics[width=0.7\textwidth]{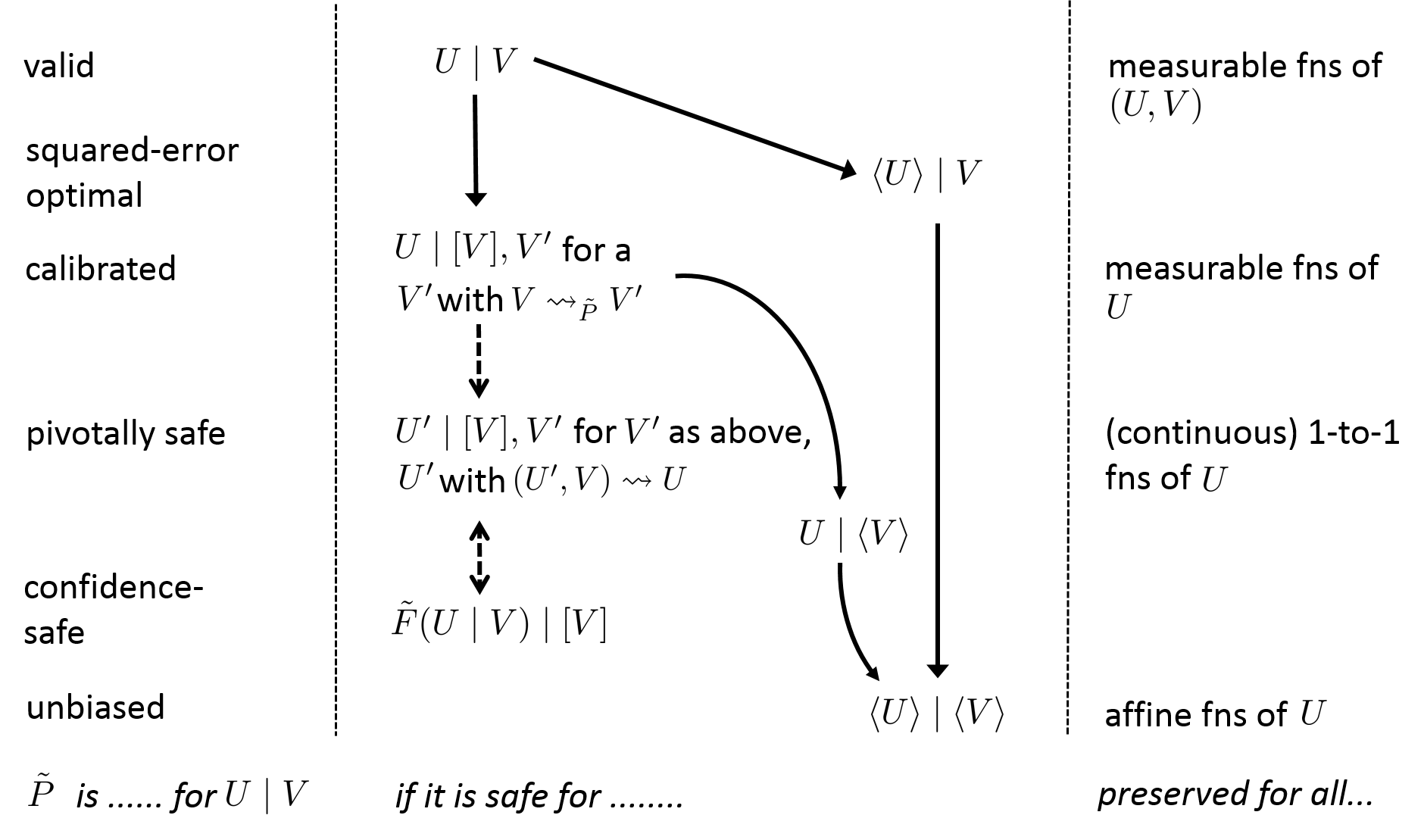}}
\caption{\small A Hierarchy of Relations for $\prag{P}$.  The concepts
  on the right correspond (broadly) to existing notions, whose name is
  given on the left (with the exception of $U \mid \ave{V}$, for which
  no regular name seems to exist).  $A \rightarrow B$ means that
  safety of $\prag{P}$ for $A$ implies safety for $B$ --- at least,
  under some conditions: for all solid arrows, this is proven under
  the assumption of $V$ with countable range (see underneath
  Proposition~\ref{prop:newstart}). For the dashed arrows, this is
  proven under additional conditions (see Theorem~\ref{thm:confidence} and
  subsequent remark). On the right are shown transformations on $U$
  under which safety is preserved, e.g.  if $\prag{P}$ is calibrated
  for $U| V$ then it is also calibrated for $U' \mid V$ for every $U'$
  with $U\determines U'$ (see remark underneath
  Theorem~\ref{thm:confidence}).  Weakening the conditions for the
  proofs and providing more detailed interrelations is a major goal
  for future work, as well as investigating whether the hierarchy has
  a natural place for {\em causal\/} notions, such as $\prag{P}(U \mid
  \text{\sc do}(v))$ as in Pearl's (\citeyear{Pearl09})
  do-calculus. \label{fig:figure}}
\end{figure}
\begin{example}{\bf \ [Dilation]}\label{ex:dilation}
  A DM has to make a prediction or decision about random variable $U  \in \cU = \{0, 1\}$ given the value of $V \in \cV= \{0,1\}$. She
  knows that the marginal probability $P(U=1) = 0.9$; she suspects
  that $U$ may depend on $V$, but has no idea whether $U$ and $V$ are
  positively or negatively correlated or how strong the correlation
  is. She may thus model her uncertainty as the set $\cP^*$ of {\em
    all\/} distributions $P$ on $\cZ = \cU \times \cV$ that satisfy
\begin{equation}\label{eq:firstconstraint}
P(U=1) = \sum_{v \in \cV} P(U=1,V=v) = 0.9.\end{equation}
Given that $V=1$, what
  should she predict for $U$? A standard answer in imprecise
  probability \citep{Walley91} is to pointwise condition the set
  $\cP^*$, leading one to adopt the probabilities $\cP^*(U=1 \mid V=1)
  := \{P(U=1 \mid V=1) : P \in \cP^* \}$. But this set contains {\em
    every\/} distribution on $U$, including $P(U=1 \mid V=1)= 0$ (the
  latter would obtain for the $P \in \cP^*$ with $P(U= |1- V|) =
  1$). It therefore seems that, after observing $V=1$, the DM has lost
  rather than gained information. By symmetry, the same happens after
  observing $V=0$, so {\em whatever DM observes, she loses
    information} --- a phenomenon known as {\em dilation\/}
  \citep{SeidenfeldW93}. This is intuitively
  disturbing, and it may perhaps be better to simply {\em ignore\/}
  $V$ and predict using the distribution that acts as if $U \perp V$
  and has 
\begin{equation}\label{eq:theignorer}
\prag{P}(U=1 \mid V=v) = {P}(U=1) \text{\ \ \ for all $v \in \cV$},
\end{equation}
i.e. $\prag{P}(U=1 \mid V=v) = 0.9$.  While from a purely subjective
Bayesian standpoint information is never useless and this seems silly,
it is certainly what humans often do in practice, and usually, they
get away with it \citep{Dempster68} --- for concrete examples see
\cite{GrunwaldH04}. Here is where Safe Probability comes in --- it
tells us that $\prag{P}$ is {\em safe\/} to use, in the following
simple sense: for any function $g: \cU \rightarrow \reals$, we have:
\begin{equation}\label{eq:babycalibration}
\text{for all $P \in \cP^*$, all $v \in \cV$:}\ \ \   
\Exp_{U \sim P}[ g(U)] = \Exp_{U \sim \prag{P}} [g(U) \mid V=v ].
\end{equation}
In particular, if we have a loss function $L: \cU
  \times \cA \rightarrow \reals$ mapping outcomes and actions to associated losses, then, for any action $a \in \cA$, we can plug in $g(U) := L(U,a)$ above and then we find that (assuming $\cP^*$ contains the truth):
\begin{quote}
  DM's predictions are guaranteed to be exactly as good, in
  expectation, as {\em she would expect them to be if $\prag{P}$ were
    actually `true'\/} --- even if $\prag{P}$ is not true at all. 
\end{quote}
We immediately add though that if we had a loss function $L': \cU
\times \cV \times \cA\rightarrow \reals $ which would {\em itself\/}
depend on $V$ (e.g. if $V=1$ DM is offered a different bet on $U$ than
if $V=0$) then the $\prag{P}$ based on ignoring $V$ is not safe any
more --- (\ref{eq:babycalibration}) may not hold any more, and the
actual expectation may be different from DM's. In terms of the
formalism we develop below (Definition~\ref{def:leftsafety}, \ref{def:maina}
and~\ref{def:mainb}), this will be expressed as `$\prag{P}$ is safe
for predicting with loss function $L$ but not loss function $L'$', or,
in formal notation, $\prag{P}$ is safe for
$L(\cdot,a) \mid [V]$ but not for $L'(\cdot, a) \mid [V]$. The
intuitive meaning is that DM can safely use $\prag{P}$ to make
predictions against $L$ (her predictions will be as good as she
expects) but not against $L'$. These statements will be immediate
consequences of the more general statements `$\prag{P}$ is safe for
$U \mid [V]$ but not safe for $U\mid V$'.

In some cases, we will not be able to come up with a $\prag{P}$
satisfying (\ref{eq:babycalibration}), and we have to settle for a
$\prag{P}$ that satisfies a weaker notion of safety, such as, for all
$P \in \cP^*$, all functions $g$,
\begin{equation}\label{eq:babyversion}
\Exp_{V \sim P} \left[ \Exp_{U \sim \prag{P}} \left[ g(U) \mid V\right] \right]
= E_{U \sim P}[g(U)], 
\end{equation}
which says that DM predicts as well on average as DM would expect to predict on average if $\prag{P}$ were true, even though $\prag{P}$ may not be true. 
This will be denoted as `$\prag{P}$ is safe for $U \mid \ave{V}$'; and if (\ref{eq:babyversion}) only holds for $g$  the identity (which makes no difference if 
$|\cU| = 2$, but in general it does) we have the even weaker safety for $\ave{U} \mid \ave{V}$ (Figure~\ref{fig:figure}). In Section~\ref{sec:four} we
  thus obtain five basic notions of safety, varying from weak safety,
  in an average sense, to very strong safety, safety for $U \mid V$,
  which essentially means that $\prag{P}(U \mid V)$ must be the
  correct conditional distribution.
\end{example}
In this example we used frequentist terminology, such as `correct' and
`true', and we  continue to do so in this paper. Still, a subjective
interpretation remains valid in this and future examples as well: if
the DM's real beliefs are given by the full set $\cP^*$, she can
safely act as if her belief is represented by the singleton
$\prag{P}$ as long as she also believes that her loss does not depend
on $V$.

\begin{example}{\bf \ [Calibration]} 
\label{ex:calibration}
Consider the weather forecaster on your local television
station. Every night the forecaster makes a prediction about whether
or not it will rain the next day in the area where you live. She does
this by asserting that the probability of rain is $p$, where $p \in
\{0,0.1, \ldots, 0.9,1 \}$. How should we interpret these
probabilities? The usual interpretation is that, in the long run, on
those days at which the weather forecaster predict probability $p$, it
will rain approximately $100 p \%$ of the time.  Thus, for example,
among all days for which she predicted $0.1$, the fraction of days
with rain was close to $0.1$.  A weather forecaster (DM) with this property
is said to be {\em calibrated\/} \citep{Dawid82,FosterV98}.
Like safety itself, calibration is a {\em minimal\/} requirement: for
example, a weather forecaster who predicts, each day of the year, that
the probability of rain tomorrow is $50\%$ will be approximately
calibrated in the Netherlands, but her predictions are not very useful
--- and it is easily seen that, when using a proper scoring rule,
optimal forecasts are calibrated, but calibrated forecasts can be far
from optimal.  On the other hand, in practice we often see calibrated
weather forecasters that predict well, but do not predict with
anything close to the `truth' --- their predictions depend on
high-dimensional covariates consisting of measurements of air
pressure, temperature etc. at numerous locations in the world, and it
seems quite unlikely (and, for practical purposes, unnecessary!)
that, given any specific values of these covariates, they issue the
correct conditional distribution.  While calibration is usually
defined relative to empirical data, a re-definition in terms of an
underlying set of distributions $\cP^*$ is straightforward
\citep{VovkGS05,GrunwaldH11}, and in Section~\ref{sec:calibration} we
show that the probabilistic definition of calibration has a natural
expression in terms of the safety notions introduced above:
$\prag{P}(U \mid V)$ is calibrated for $U$ if it is safe for $U \mid
[V], V'$, for {\em some\/} $V'$ with $V \determines V'$ (all notation
to be explained) --- which implies that (\ref{eq:babycalibration}) is
itself an instance of calibration.
\end{example}
\begin{example}{\bf \ [Bayesian, Fiducial and Confidence
    Distributions]}
\label{ex:confidence}
We are given a parametric probability model $\cM = \{q_{\theta} \mid
\theta \in \Theta \}$ where $\Theta \subseteq \reals^k$ for some $k
\geq 1$, each $q_{\theta}$ defines a probability density or mass
function on data $(X_1, \ldots, X_N) = X^N$ of sample size $N$, each
outcome $X_i$ taking a value in some space $\cX$.  The goal is to make
inferences about $\theta$, based on the data $X^N$ or some statistic
$S(X^N,N)$ thereof. In the common case with fixed $N=n$ and inference
based on the full data, $S(X^N,N) := X^n$, we can transfer this
statistical scenario to our setup by defining $\cP^*$ as a set of
distributions on $\cZ = \Theta \times \cX^n$. RVs $U$ and $V =
S(X^n,n) = X^n$ are then defined as, for each $z = (\theta,x^n)$,
$U(z) := \theta$ and $V(z) := x^n$. DM employs a set $\Pi$ of prior
distributions on $\Theta$, where each $\pi \in \Pi$ induces a joint
distribution $P_{\pi}$ on $\Theta \times \cX^n$ with marginal on
$\Theta$ determined by $\pi$ and, given $\theta$, density of $x^n$
given by $q_{\theta}$, so that if $\pi$ has density $p_{\pi}$, we get
the joint density $p_{\pi}(\theta, x^n) = p_{\pi}(\theta) \cdot
q_{\theta}(x^n)$.  We set $\cP^* := \{ P_{\pi} : \pi \in \Pi \}$ to be
the set of all such joint distributions.  In the special case in which
DM really is a $100\%$ subjective Bayesian who believes that a single
prior $\pi$ captures all uncertainty, we have that $\cP^* = \{
P_{\pi} \}$ contains just a single joint parameter-data distribution,
and we are in the standard Bayesian scenario. Then DM can set
$\prag{P}(\theta \mid X^n) := P_{\pi}(\theta \mid X^n)$, the standard
posterior, and any type of inference about $\theta$ is safe relative
to $\cP^*$. Here we focus on another special case, in
which $\Pi$ contains exactly one density for each $\theta \in \Theta$,
namely the degenerate distribution putting all its mass on
$\theta$. We denote this distribution by $P_{\theta}$ and notice that
then $\cP^* = \{P_{\theta} : \theta \in \Theta \}$, with
$P_{\theta}(\Theta = \theta) = 1$, and for any measurable set $\cA$, $P_{\theta}(X^n \in \cA)$
determined by density $p_{\theta}$, satisfying
$$p_{\theta}(x^n) = p_{\theta}(x^n \mid \Theta = \theta) =  q_{\theta}(x^n).
$$ 
Still, any choice of pragmatic distribution
$\prag{P}(U \mid V) = \prag{P}(\theta \mid X^n)$ can be interpreted as
a distribution on $U \equiv \Theta$ given the data $X^n$, analogous to
a Bayesian posterior. In Section~\ref{sec:continuous} we investigate
how one can construct distributions $\tilde{P}$ of this kind that are
safe for inference about {\em confidence intervals}. for simplicity we
restrict ourselves to the 1-dimensional case, for which we find that
the construction we provide leads to $\tilde{P}$ that are
confidence-safe, written in our notation as `safe for
$\tilde{F}(U|V) \mid [V]$', with $\tilde{F}$ being the CDF (cumulative
distribution function) of $\tilde{P}(U|V)$. Confidence safety is
roughly the same as coverage \cite{Sweeting01}: it means that the `true'
probability that $\theta$ is contained in a particular type of
$\alpha$-credible sets (sets with `posterior' probability $\alpha$
given the data $V$), is equal to $\alpha$.

The $\tilde{P}$ we construct are essentially equivalent to the {\em
  confidence distributions\/} of \citep{SchwederH02}, that were
designed with the explicit goal of having good confidence properties;
they also often coincide with Fisher's \citeyear{Fisher30} fiducial
distributions, which in later work \citep{Fisher35} he started
treating as ordinary probability distributions that could be used
without any restrictions. This cannot be right (see e.g. \cite[page
514]{Hampel06}), but the question has always remained how a
probability calculus for fiducial distributions could be derived that
incorporates the right restrictions. Our work provides a step in this
direction, in that we show how such $\tilde{P}$ snugly fit into our
general framework: confidence safety is a strictly weaker property
than calibration, and has again a natural representation in terms of
the $\ave{U} \mid \ave{V}$ notation mentioned above. Moreover, it is a
special case of {\em pivotal safety\/} which also has repercussions in
quite different contexts --- see Example~\ref{ex:montyhall}.
\end{example}
The example illustrates two important points:
\begin{enumerate}
\item In some cases the literature suggests some method for
  constructing a pragmatic  $\prag{P}$. An example is the latent Dirichlet
  allocation  model \citep{blei2003latent} mentioned above, in
  which data $V$ are text corpora, $\cP^*$, not explicitly given, is a complicated set of realistic distributions
  over $V$ under which data are non-i.i.d., and the literature
  suggests to take $\prag{P}(U \mid V)$ as the Bayesian posterior for
  a cleverly designed i.i.d. model.
\item In other cases, DM may want to construct a $\prag{P}$ herself.
  In Example~\ref{ex:dilation}, the safe $\prag{P}$ was obtained by
  replacing an (unknown) conditional distribution with a (known)
  marginal --- a special case of what was called $\cC$-conditioning by
  \cite{GrunwaldH11}.  Marginal distributions and distributions that
  ignore aspects of $V$ play a more central role in this construction
  process: they also do in the confidence construction mentioned above,
  where one sets $\prag{P}(U \mid V)$ equal to a distribution such
  that $\prag{P}(U' \mid V)$, where $U'$ is some auxiliary random
  variable (a {\em pivot\/}), becomes independent of $V$. For the
  original RV $U$ though, in the dilation example, DM acts as if $U$
  and $V$ are independent even though they may not be; in the
  confidence distribution example, DM acts in a `dual' manner, namely
  as if $U$ and $V$ are dependent, even though under $\cP^*$ they are
  not --- which is fine, as long as her conclusions are {\em safe}.
\end{enumerate}
\begin{example}{\bf \ [Event-Based Conditioning and Pivotal Safety via 
Monty Hall]} 
\label{ex:montyhall}
More generally, we may look at safety for pragmatic distributions
$\prag{P}$ that condition on events rather than random variables. To
illustrate, consider the Monty Hall Problem
\citep{vossavant90,Gill11}: suppose that you're on a game show and
given a choice of three doors $\{1,2,3\}$ Behind one is a car; behind
the others are goats.  You pick door $1$.  Before opening door $1$,
Monty Hall, the host opens one of the other two doors, say, door $3$
which has a goat.  He then asks you if you still want to take what's
behind door $1$, or to take what's behind door $2$ instead.  Should
you switch?  You may assume that initially, the car was equally likely
to be behind each of the doors and that, after you go to door $1$,
Monty will always open a door with a goat behind.  Basically you
observe either the event $\{1,3\}$ (if Monty opens door $2$) or
$\{1,2\}$ (if Monty opens $3$). You can then calculate your optimal
decision according to some distribution $\prag{P}(\{1\} \mid \cE)$,
where $\cE \in \{ \{1,3\}, \{1,2\} \}$ is the event you
observed. Naive conditioning suggests to take
$P(\{1 \} \mid \{1, 2 \}) = P(\{1 \} \mid \{1, 3 \}) = (1/3) / (2/3) =
1/2$, and it takes a long time to convince most people that this is
wrong --- but, if DM's would adhere to safe probability, then no
convincing and explanation would be needed: translation of the example
into our `safety' setting immediately shows, without any further
thinking about the problem, that this choice of $\prag{P}$ is {\em
  unsafe}, under all notions of safety we consider!
(Section~\ref{sec:general}).  

Another aspect of the Monty Hall problem is that, in most analyses
that are usually viewed as `correct', one implicitly assumes that the
quiz master flips a {\em fair\/} coin to decide whether to open door
$2$ or $3$ if you choose door 1 so that he has a choice. There have
been heated discussions (e.g. on wikipedia talk pages) about whether
this assumption is justified. In Example~\ref{ex:montyhallb} we show
that the $\prag{P}$ which assumes a fair coin flip by Monty is an
instance of a {\em pivotally safe\/} pragmatic distribution. These
have the properties that for many loss functions (including $0/1$-loss
as in Monty Hall), they lead one to making optimal decisions. Thus, while assuming a fair coin flip may be wrong, it is still {\em harmless\/} to base one's decisions upon it. 
\end{example}
\paragraph{Overview of the Paper}
In Section~\ref{sec:discrete}, we treat the case of countable space
$\cZ$, defining the basic notions of safety in
Section~\ref{sec:four} (where we return to dilation), and showing how
calibration can be cleanly expressed using our notions in
Section~\ref{sec:calibration}. In Section~\ref{sec:continuous} we
extend the setting to general $\cZ$, which is needed to handle the
case of confidence safety (Section~\ref{sec:confidence}), pivots
(Section~\ref{sec:pivots}) and squared error optimality, where we
observe continuous-valued random variables. Section~\ref{sec:general}
briefly discusses non-numerical observations as well as probability
updates that cannot be viewed as conditional distributions. We end
with a discussion of further potential applications of safety as well
as open problems. Proofs and further technical details are delegated
to the appendix.

\section{Basic Definitions for Discrete Random
  Variables}\label{sec:discrete}
For simplicity, we introduce our basic notions only considering
countable $\cZ$, which allows us to sidestep measurability issues
altogether.
Thus below, $\cZ$ is countable; we treat
the general case in Section~\ref{sec:continuous}.
\subsection{Concepts and notations regarding distributions on $\cZ$} 
\label{sec:pstar}
We define a random variable (abbreviated to RV) to be any function $X:
\cZ \rightarrow \reals^k$ for some $k > 0$.
Thus RVs can be multidimensional (i.e. what is usually
called `random vector').
By an `$\cY$-valued RV' or
simply `generalized RV' we mean any function mapping $\cZ$ to an arbitrary set
$\cY$.
For two RVs $U = (U_1, \ldots, U_{k_1}), V= (V_1, \ldots, V_{k_2})$
where $U_j$ and $V_j$ are 1-dimensional random variables, we define
$(U,V)$ to be the RV with components $(U_1, \ldots, U_{k_1}, V_1,
\ldots, V_{k_2})$.

For any generalized RVs $U$ and $V$ on $\cZ$ and function $f$ we write
$U \fdetermines V$ if for all $z \in \cZ$, $V(z) = f(U(z))$. We write
$U \determines V$ (``$U$ determines $V$'', or equivalently ``$U$ is a
{\em coarsening\/} of $V$'') if there is a function $f$ such that $U
\fdetermines V$.  
We write $U \onetoones V$ if $U \determines V$
and $V \determines U$.  
For two GRVs $U$ and $V$ we write $U \equiv V$ if they define the same
function on $\cZ$, and for a distribution $P \in \cZ$ we write $U
{=}_P V$ if $P(U=V) =
1$. 
We write $U \fPdetermines{f}{P} V$ if $P(\{z \in \cZ: V(z) = f(U(z)) \}) = 1$, and $U \Pdetermines{P} V$ if there exists some $f$ for which this holds. 
Clearly $U \determines V$ implies that for all distributions $P$ on $\cZ$, 
$U \Pdetermines{P} V$, but not vice versa. 
Let $S: \cZ \rightarrow \cod{S}$ be a function on $\cZ$. The {\em
  range\/} of $S$, denoted $\range{S}$, the {\em support\/} of $S$
under a distribution $P$, and the range of $S$ given that another
function $T$ on $\cZ$ takes value $t$, are denoted as
\begin{align}\label{eq:range}
& \range{S} := 
\{s \in \cod{S}: s= S(z) \text{\ for some $z \in \cZ$} \} \ \ ; \ \  
\support_{P}(S) := \{s \in \cod{S}: P(S=s) > 0 \}, \nonumber \\
& \range{S\mid T=t} = \{ s \in \cod{S}: s=S(z) \text{\ for some $z \in \cZ$ with $t= T(z)$} \}  \end{align}
where we note that $\support_P(S) \subseteq \range{S}$, with equality
if $S$ has full support.
%

For a distribution $P$ on $\cZ$, and $\cU$-valued RV $U$, we write
$P(U)$ as short-hand to denote the distribution of $U$ under $P$
(i.e. $P(U)$ is a probability measure).

We generally omit double brackets, i.e. if we write $P(U,W)$ for RVs
$U$ and $W$, we really mean $P( R)$ where $R$ is the RV $(U,W)$,

Any generalized RV that maps all $z \in \cZ$ to the same constant is
called {\em trivial}, in particular the RV ${\bf 0}$ which maps all $z
\in \cZ$ to $0$.  For an event $\cE \subset \cZ$, we define the {\em
  indicator random variable\/} $\indicator_{\cE}$ to be $1$ if $\cE$
holds and $0$ otherwise.

\paragraph{Conditional Distributions as Generalized RVs}
For given distribution on $\cZ$ and generalized RVs $V$ and $W$, we
denote, for all $v \in \support_P(V)$, $P \mid V=v$ as the conditional
distribution on $\cZ$ given $V=v$, in the standard manner. We further
define $(\cP^* \mid W=w) := \{ (P \mid W=w) : P \in \cP^*, w \in
\support_{P}(W) \}$ to be the set of distributions on $\cZ$ that can
be arrived at from $\cP$ by conditioning on $W=w$, for all $w$
supported by some $P \in \cZ$.

We further denote, for all $v \in \support_P(V)$, $P(U \mid V=v)$ as
the conditional distribution of $U$ given $V=v$, defined as the
distribution on $U$ given by $P(U = u \mid V=v) := P(U=u,V=v)/P(V=v)$
(whereas $P \mid V=v$ is defined as a distribution on $\cZ$, $P(U \mid
V=v)$ is a distribution on the more restricted space $\range{U}$).

Suppose DM is interested in predicting RV $U$ given RV $V$ and does
this using some conditional distribution $P(U \mid V=v)$ (usually this
$P$ will be the `pragmatic' $\prag{P}$, but the definition that
follows holds generally). Adopting the standard convention for
conditional expectation, we call any function from $\range{V}$ to the
set of distributions on $U$ that coincides with $P(U \mid V=v)$ for
all $v \in \support_P(V)$ a {\em version\/} of the conditional
distribution $P(U \mid V)$. If we make a statement of the form `$P(U
\mid V)$ satisfies ...', we really mean `every version of $P(U \mid
V)$ satisfies...'. We thus treat $P(U \mid V)$ as a $\cE$-valued
random variable where $\cE = \{ {P}(U\mid V=v) : v \in \range{V}
\}$, where, for all $z \in \cZ$ with $P(V = V(z)) > 0$, $P(U \mid
V)(z) := P(U \mid V= V(z))$, and $P(U \mid V)(z)$ set to an arbitrary
value otherwise.

\paragraph{Unique and Well-Definedness}
Recall that DM starts with a set $\cP^*$ of distributions on $\cZ$
that she considers the right description of her uncertainty. 
She will predict sume RV $U$ given some generalized RV $V$ using a {\em pragmatic\/} distribution $\prag{P}$. 

For RV $U: \cZ \rightarrow \reals^k$ and generalized RV $V$, we say
that, for given distribution $P'$ on $\cZ$, ${P}'(U \mid V)$ is {\em
 essentially uniquely defined\/} (relative to $\cP^*$) if for all $P \in \cP^*$,
$\support_P(V) \subseteq \support_{P'}(V)$ (so that $P$-almost surely
$V$ takes value $v$ with $P'(V=v) > 0$). We use this definition both
for $P' \in \cP^*$ and for $P' = \prag{P}$; note that we always {\em
  evaluate\/} whether $P'$ is uniquely defined under distributions in
the `true' $\cP^*$ though.

We say that $\Exp_{{P}'}[U \mid V]$ is well-defined if,
writing $U= (U_1, \ldots, U_k)$, and,
$U^+_j = \max \{U_j, 0 \}$, $U^-_j = \max \{ - U_j, 0 \}$, we have,
for $j= 1..k$, either $\Exp_{{P}'}[U^+_j \mid
V] < \infty$ with $P$-probability 1, or $\Exp_{{P}'}[U^-_j \mid V] <
\infty$ with $P$-probability 1. This is a very weak requirement that
ensures that calculating expectations never involves the operation
$\infty - \infty$, making all expectations well-defined.

\paragraph{The Pragmatic Distribution $\prag{P}$} 
We assume that DM makes her predictions based on a probability
distribution $\prag{P}$ on $\cZ$ which we generally refer to as the
{\em pragmatic distribution}.  In practice, DM will usually be
presented with a decision problem in which she has to predict some
fixed RV $U$ based on some fixed RV $V$, and then she is only
interested in the conditional distribution $\prag{P}(U \mid V)$, and
for some other RVs $U'$ and $V'$, $\prag{P}(U' \mid V')$ may be left
undefined. In other cases she only may want to predict the expectation
of $U$ given $V$ --- in that case she only needs to specify
$\Exp_{\prag{P}}[U \mid V]$ as a function of $V$, and all other
details of $\prag{P}$ may be left unspecified. In
Appendix~\ref{app:partial} we explain how to deal with such {\em
  partially specified\/} $\prag{P}$. In the main text though, for
simplicity we assume that $\prag{P}$ is a fully-specified distribution
on $\cZ$; DM can fill up irrelevant details any way she likes. The
very goal of our paper being to restrict $\prag{P}$ to making `safe'
predictions however, DM may come up with $\prag{P}$ to predict $U$
given $V$ and there may be many RVs $U'$ and $V'$ definable on the
domain such that $\prag{P}(U' \mid V')$ has no bearing to $\cP^*$ and
would lead to terrible predictions; as long as we make sure that
$\prag{P}$ is not used for such $U'$ and $V'$ --- which we will ---
this will not harm the DM.

\subsection{The  Basic Notions of Safety}
\label{sec:four}
All our subsequent notions of `safety' will be constructed in terms of
the following first, simple definitions. 
\begin{definition}\label{def:leftsafety}
  Let $\cZ$ be an outcome space and $\cP^*$ be a set of distributions
  on $\cZ$, let $U$ be an RV and $V$ be a generalized RV on $\cZ$, and
  let $\prag{P}$ be a distribution on $\cZ$.  We say that $\prag{P}$
  is {\em safe\/} for $\ave{U}_{\we} \mid \wksafe{V}_{\we}$ (pronounced
  as `$\prag{P}$ is safe for {\em predicting\/} $\ave{U}$ {\em
    given\/} $\wksafe{V}$'), if
\begin{equation}\label{eq:startsafer}
\text{for all $P \in \cP^*$}: 
\inf_{v \in \support_{\prag{P}}(V)} \Exp_{\prag{P}} [U | V=v] \leq 
\ \Exp_P[U] \leq \sup_{v \in \support_{\prag{P}}(V)} \Exp_{\prag{P}} [U | V=v].
\end{equation}
We say that $\prag{P}$ is {\em
    safe\/} for $\ave{U}_{\we} \mid \ave{V}_{\we}$, if 
\begin{equation}\label{eq:startsafea}
\text{for all $P \in \cP^*$}: \ \Exp_P[U] = \Exp_{P}[ \Exp_{\prag{P}} [U | V]].
\end{equation}
We say that $\prag{P}$ is {\em
    safe\/} for $\ave{U}_{\we} \mid \cali{V}$, if (\ref{eq:startsafer}) holds with both inequalities replaced by an equality, i.e. for all $v \in \support_{\tilde{P}}(V)$, 
\begin{equation}\label{eq:startsafes}
\text{for all $P \in \cP^*$}: \ \Exp_P[U] =  \Exp_{\prag{P}} [U | V = v].
\end{equation}
\end{definition}
In this definition, as in all definitions and results to come,
whenever we write `$\langle$ statement $\rangle$' we really mean `all
conditional probabilities in the following statement are essentially uniquely
defined, all expectations are well-defined, and $\langle$ statement
$\rangle$'.  Hence, (\ref{eq:startsafea}) really means `for all
$P \in \cP^*$, $\prag{P}(U|V)$ is essentially uniquely defined,
$\Exp_{\prag{P}} [U | V]$, $\Exp_P[U]$, and
$\Exp_{P}[ \Exp_{\prag{P}} [U | V]]$ are well-defined, and the latter
two are equal to each other'.  Also, when we wrote $\prag{P}$ is safe
for $\ave{U}_{\we} \mid \ave{V}_{\we}$, we really meant that it is
safe for $\ave{U}_{\we} \mid \ave{V}_{\we}$ {\em relative to the given
  $\cP^*$}; we will in general leave out the phrase `relative to
$\cP^*$', whenever this cannot cause confusion.

To be fully clear about notation, note that in double expectations
like in (\ref{eq:startsafea}), we consider the right random variable
to be bound by the outer expectation; thus it can be rewritten in any
of the following ways:
\begin{align*}
  \Exp_{U \sim P}[U] & =   \Exp_{V \sim P}  \Exp_{U \sim \prag{P} \mid V}[U] \\
  \Exp_{V \sim P} \Exp_{U \sim P \mid V} [U] & = \Exp_{V \sim P}
  \Exp_{U \sim \prag{P} \mid V} [U] \\ 
\sum_{u \in \range{U}} P(U=u) \cdot u & = \sum_{v \in
      \range{V}} P(V=v) \cdot \sum_{u \in \range {U}} \prag{P}(U=u \mid
    V=v) \cdot u,
\end{align*}
where the second equality follows from the tower property of
conditional expectation.  

\paragraph{Towards a Hierarchy} It is immediately seen that, if
$\tilde{P}$ is safe for $\ave{U} \mid \cali{V}$, then it is also safe
for $\ave{U} \mid \ave{V}$, and if it is safe for $\ave{U} \mid
\ave{V}$, then it is also safe for $\ave{U} \mid \wksafe{V}$.  Safety
for $\ave{U} \mid \wksafe{V}$ is thus the weakest notion --- it allows
a DM to give valid upper- and lower-bounds on the actual expectation
of $U$, by quoting $\sup_{v \in \support_{\tilde{P}}(V)}
\Exp_{\tilde{P}}[U \mid V=v]$ and $\inf_{v \in
  \support_{\tilde{P}}(V)} \Exp_{\tilde{P}}[U \mid V=v]$,
respectively, but nothing more. It will hardly be used here, except
for a remark below Theorem~\ref{thm:confidence}; it plays an important
role though in applications of safety to hypothesis testing, on which
we will report in future work.

Safety for $\ave{U} \mid \ave{V}$ evidently bears relations to {\em
  unbiased\/} estimation: if $\prag{P}$ is safe for $\ave{U} \mid
\ave{V}$, i.e. (\ref{eq:startsafea}) holds, then we can think of
$\Exp_{\prag{P}} [U | V]$ as an unbiased estimate, based on observing
$V$, of the random quantity $U$ (see also Example~\ref{ex:supernormal}
later on).  Safety for $\ave{U} \mid \cali{V}$ implies that all
distributions in $\cP^*$ agree on the expectation of $U$ and that
$\Exp_{\tilde{P}}[U \mid V=v]$ is the same for (essentially) all
values of $v$, and is thus a much stronger notion.

\begin{example}{\bf [Dilation: Example~\ref{ex:dilation}, Cont.]}
\label{ex:dilationb}
The first application of definition (\ref{eq:startsafea}) was already
given in Example~\ref{ex:dilation}, where we used a $\prag{P}$ that
ignored $V$ and was safe for $\ave{U} \mid \ave{V}$ and $\ave{U} \mid
\cali{V}$, as we see from (\ref{eq:babyversion}) with
$g$ the identity.  Let us extend the example, replacing $\cU =
\{0,1\}$ in that example by $\cU = \{0,1,2\}$, with $\cP^*$ again
defined as the set of all distributions satisfying
(\ref{eq:firstconstraint}) and $\prag{P}$ defined by, for $v \in
\{0,1\}$, $\prag{P}(U=1 \mid V= v) = 0.9$, $\prag{P}(U=2 \mid V = v) =
0.09$. Then $\prag{P}$ would still be safe for $\ave{\indicator_{U=1}}
\mid \ave{V}$, but not for $\ave{U} \mid \ave{V}$: $\cP^*$ contains a
distribution whose marginal distribution $P(U=2) = 0$, and
(\ref{eq:startsafea}) would not hold for that distribution.
\end{example}
Comparing the `safety condition' (\ref{eq:babyversion}) in
Example~\ref{ex:dilation} to (\ref{eq:startsafea}) in
Definition~\ref{def:leftsafety} we see that
Definition~\ref{def:leftsafety} only imposes a requirement on
expectations of $U$ whereas (\ref{eq:babyversion}) imposed a
requirement also on RVs $U'$ equal to functions $g(U)$ of $U$.  For
$\cU$ with more than two elements as in Example~\ref{ex:dilationb}
above, such a requirement is strictly stronger. We now proceed to
define this stronger notion formally.
\begin{definition}
\label{def:maina} 
Let $\cZ, \cP^*$, $U, V$ and $\prag{P}$ be as above.  We say that
$\prag{P}$ is {\em safe\/} for $U \mid_{\we} \wksafe{V}$ if for all RVs
$U'$ with $U \determines U'$, $\prag{P}$ is safe for $\ave{ U'}
\mid_{\we} \wksafe{V}$. \\ Similarly, $\prag{P}$ is safe for $U \mid
\ave{V}$ if for all RVs $U'$ with $U \determines U'$, $\prag{P}$ is
safe for $\ave{ U'} \mid_{\we} \ave{V}$, and $\prag{P}$ is safe for $U
\mid \cali{V}$ if for all RVs $U'$ with $U \determines U'$, $\prag{P}$
is safe for $\ave{ U'} \mid_{\we} \cali{V}$.
\end{definition}
We see that safety of $\prag{P}$ for $U \mid \cali{V}$ implies that
$\Exp_{\tilde{P}}[g(U) \mid V=v]$ is the same for all values of $v$ in
the support of $\tilde{P}$, and all functions $g$ of $U$. This can
only be the case if $\tilde{P}(U \mid V)$ {\em ignores\/} $V$,
i.e. $\tilde{P}(U \mid V=v) = \tilde{P}(U)$, for all supported $v$.
We must then also have that, for all $v \in \support_{\tilde{P}}(V)$,
that $\tilde{P}(U) = P(U)$, which means that all distributions in
$\cP^*$ agree on the marginal distribution of $U$, and $\tilde{P}(U)$
is equal to this marginal distribution.  Thus, $\prag{P}$ is safe for
$U \mid \cali{V}$ iff it is {\em marginally valid}. A prime example of
such a $\prag{P}(U \mid V)$ that ignores $V$ and is marginally correct
is the $\prag{P}(U \mid V)$ we encountered in
Example~\ref{ex:dilation}.

To get everything in place, we need a final definition. 
\begin{definition}
\label{def:mainb} 
Let $\cZ, \cP^*$, $U, V$ and $\prag{P}$ be as above, and let $W$ be
another generalized RV.  
\begin{enumerate} 
\item We say that $\prag{P}$ is {\em safe\/} for
$\ave{U} \mid_{\we} \wksafe{V}, W$ if for all $w \in
\support_{\prag{P}}(W)$, $\prag{P} \mid W=w$  
is safe for $\ave{U} \mid \wksafe{V}$ relative to $\cP^* \mid W=w$. 
We say that $\prag{P}$ is {\em safe\/} for
${U} \mid_{\we} \wksafe{V}, W$  if for all RVs
$U'$ with $U \determines U'$, $\prag{P}$ is safe for  $\ave{U'} \mid_{\we} \wksafe{V}, W$. 
\item The same definitions apply with $\wksafe{V}$ replaced by $\ave{V}$ and $\cali{V}$. 
\item We say that $\prag{P}$ is {\em safe\/} for $\ave{U} \mid_{\we}
  W$ if it is safe for $\ave{U} \mid_{\we} \wksafe{{\bf 0}}, W$; it is
  safe for $U \mid_{\we} W$ if it is safe for ${U} \mid_{\we}
  \wksafe{{\bf 0}}, W$.
\end{enumerate}
\end{definition}
These definitions simply say that safety for `$.. | .., W$' means that
the space $\cZ$ can be partitioned according to the value taken by
$W$, and that for each element of the partition (indexed by $w$) one
has `local' safety given that one is in that element of the partition.

Proposition~\ref{prop:newstart} gives  reinterpretations of some of the notions above. 
The first one, (\ref{eq:startsafec}) will mostly be useful for the proof of other
results; the other three serve to make the original definitions more  transparent: \begin{proposition}{\bf \ [Basic Interpretations of Safety]}\label{prop:newstart}
  Consider the setting above. We have:
\begin{enumerate}
\item $\prag{P}$ is safe for
$U \mid_{\we} \ave{V}$ iff for all $P \in \cP^*$, there exists a  
distribution
  $P'$ on $\cZ$ with for all $(u,v) \in \range{(U,V)}$, 
$P'(U=u , V=v) = \prag{P}(U=u \mid V=v) 
\cdot P(V=v)$, that satisfies
\begin{equation}\label{eq:startsafec}
P'(U) =  P(U).
\end{equation}
\item $\prag{P}$ is safe for $\ave{U} \mid_{\we} V$ iff for all $P \in \cP^*$, 
\begin{equation}
\label{eq:startsafeb}
\Exp_{P}[U \mid V] =_{P} \Exp_{\prag{P}}[ U \mid V].
\end{equation} 
\item $\prag{P}$ is safe for $U \mid V$ iff for all $P \in \cP^*$,
\begin{equation}\label{eq:guaranteed}
{P}(U\mid V)=_{P} \prag{P}(U \mid V).
\end{equation}
\item 
$\prag{P}$ is safe for $U \mid [V],W$ iff for all $P \in \cP^*$,
\begin{equation}\label{eq:calibrationagain}
P(U \mid W) =_P \prag{P}(U \mid V,W).
\end{equation}
\end{enumerate}\end{proposition}
Together with the preceding definitions, this proposition establishes
the arrows in Figure~\ref{fig:figure} from $U \mid V$ to $\ave{U} \mid
V$, from $\ave{U} \mid V$ to $\ave{U} \mid \ave{V}$ and from $U \mid
\ave{V}$ to $\ave{U} \mid \ave{V}$. The remaining arrows will be
established by Theorem~\ref{thm:first} and~\ref{thm:confidence}.

Note that (\ref{eq:calibrationagain}) says that $\tilde{P}$ is safe
for $U \mid {\cali V}, W$ if $\prag{P}$ ignores $V$ {\em given\/} $W$,
i.e.  according to $\prag{P}$, $U$ is conditionally independent of $V$
given $W$. Thus, $\prag{P}$ can be safe for $U \mid [V],W$ and still
$\prag{P}(U\mid V)$ may depend on $V$; the definition only requires
that $V$ is ignored once $W$ is given.

(\ref{eq:guaranteed}) effectively expresses that $\prag{P}(U \mid V)$
is valid (a frequentist might say `true') for predicting $U$ based on
observing $V$, where as always we assume that $\cP^*$ itself correctly
describes our beliefs or potential truths (in particular, if $\cP^* =
\{P\}$ is a singleton, then any $\prag{P}(U \mid V)$ which coincides
a.s. with $P(U \mid V)$ is automatically valid). Thus, `validity for
$U \mid V$', to be interpreted as {\em $\prag{P}$ is a valid
  distribution to use when predicting $U$ given observations of $V$\/}
is a natural name for safety for $U\mid V$ . We also have a natural
name for safety for $\ave{U} \mid V$: for 1-dimensional $U$,
(\ref{eq:startsafeb}) simply expresses that all distributions in
$\cP^*$ agree on the conditional expectation of $U \mid V$, and that
$\Exp_{\prag{P}}[U\mid V]$ is a version of it.  which implies (see
e.g. \cite{Williams91}) that, with the function $g(v) :=
\Exp_{\prag{P}}[U\mid V=v]$,\begin{equation} \Exp_{(U,V) \sim P}[(U -
  g(V))^2] = \min_{f} \Exp_{(U,V) \sim P}[(U - f(V))^2],
\end{equation}
the minimum being taken over all functions from $\range{V}$ to
$\reals$. This means that $\prag{P}$ encodes the {\em optimal\/}
regression function for $U$ given $V$ and hence suggests the name {\em squared-error optimality}. Summarizing the names we encountered (see Figure~\ref{fig:figure}): 
\begin{definition}{\bf [(Potential) Validity, Squared Error-Optimality, Unbiasedness, Marginal Validity]}\label{def:names}
  If $\prag{P}$ is safe for $U \mid_{\we} V$,
  i.e. (\ref{eq:guaranteed}) holds for all $P \in \cP^*$, then we also
  call $\prag{P}$ {\em valid\/} for $U \mid V$ (again, pronounce as
  `valid for {\em predicting\/} $U$ {\em given\/} $V$'). If
  (\ref{eq:guaranteed}) holds for {\em some\/} $P \in \cP^*$, we call
  $\prag{P}$ {\em potentially valid\/} for $U \mid V$.  If $\prag{P}$
  is safe for $\ave{U} \mid_{\we} V$, we call $\prag{P}$ {\em squared
    error-optimal\/} for $U \mid V$. If $\prag{P}$ is safe for
  $\ave{U} \mid_{\we} \ave{V}$, we call $\prag{P}$ {\em unbiased\/}
  for $U \mid V$. If $\prag{P}$ is safe for $\ave{U} \mid [V]$, we say
  that it is {\em marginally valid\/} for $U \mid V$.
\end{definition}
It turns out that there also is a natural name for safety for $U \mid
\cali{V}, W$ whenever $V \determines W$.  The next example reiterates
its importance, and the next section will provide the name: {\em
  calibration}.
\begin{example}
  Suppose $\tilde{P}$ is safe for $U \mid \cali{V_1}, V_2$. From
  Proposition~\ref{prop:newstart}, (\ref{eq:calibrationagain}) we see
  that this means that for all $P \in \cP^*$, all
  $v_1,v_2 \in \support_P(V_1,V_2)$, that
\begin{equation}\label{eq:validityrevisited}
\Exp_{P}[U' \mid V_2= v_2] = \Exp_{\prag{P}}[U' \mid V_1 = v_1, V_2 = v_2],
\end{equation}
The special case with $V_2 \equiv {\bf 0}$ has already been
encountered in Example~\ref{ex:dilation}, (\ref{eq:babycalibration}).
As discussed in that example, for $V_2 \equiv {\bf 0}$, (\ref{eq:validityrevisited}) 
expresses our basic interpretation of
safety that {\em predictions based on $\prag{P}$ will always be as
  good, in expectation, as the DM who uses $\prag{P}$ expects them to
  be}. Clearly this continues to be the case if (\ref{eq:validityrevisited}) holds for some nontrivial $V_2$.
\end{example}
\subsection{Calibration Safety}\label{sec:calibration}
In this section, we show that {\em calibration}, as informally defined
in Example~\ref{ex:calibration}, has a natural formulation in terms of
our safety notions. We first define calibration formally, and then, in
our first main result, Theorem~\ref{thm:first}, show how being
calibrated for predicting $U$ based on observing $V$ is essentially
equivalent to being safe for $U \mid [V], V'$ for some types of $V'$ that
need not be equal to $V$ itself, including $V' \equiv {\bf 0}$. Thus, we now effectively unify the ideas underlying Example~\ref{ex:dilation} (dilation)
and Example~\ref{ex:calibration} (calibration).
 
Following \cite{GrunwaldH11} we
define calibration directly in terms of distributions rather than
empirical data, in the following way: 
\begin{definition}{\bf \ [Calibration]\ } \label{def:calibration} Let
  $\cZ$, $\cP^*$, $U$, $V$ and $\prag{P}$ be as above.  We say that
  $\prag{P}$ is {\em calibrated\/} (or {\em calibration--safe\/}) for
  $\ave{U} \mid V$ if for all $P \in \cP^*$, all ${\mu} \in \{
  \Exp_{\prag{P}}[U \mid V= v]: v \in \support_P( V )\}$,
\begin{equation}\label{eq:weakcalibrationsafe}
\Exp_{P}[U  \mid \; \Exp_{\prag{P}}[U \mid V] = {\mu}\;] =  {\mu}.
\end{equation}
We say that $\prag{P}$ is
  {\em calibrated\/} for ${U} \mid V$ if  for all $P \in \cP^*$, all
  ${p} \in \{ \prag{P}(U \mid V= v): v \in \support_P( V )\}$,
\begin{equation}\label{eq:calibrationsafe}
P(U  \mid \; \prag{P}(U \mid V) = {p}\;) =  {p}
\end{equation}
\end{definition}
Hence, calibration (for $U$) means that given that a DM who uses
$\prag{P}$ {\em predicts\/} a specific distribution for $U$, the {\em
  actual\/} distribution is indeed equal to the predicted
distribution. Note that here we once again treat
$\prag{P}(U \mid V)$ as a generalized RV. 

In practice we would want to weaken Definition~\ref{def:calibration}
to allow some slack, requiring the $\mu$ (viz. $p$) inside the
conditioning to be only within some $\epsilon > 0$ of the $\mu$
(viz. $p$) outside, but the present idealized definition is sufficient
for our purposes here. Note also that the definition refers to a
simple form of calibration, which does not involve selection rules
based on past data such as used by, e.g., \cite{Dawid82}.

We now express calibration in terms of our safety
notions. We will only do this for the `full distribution'--version (\ref{eq:calibrationsafe}); a similar result can be established for the average-version.

\begin{theorem}\label{thm:first} 
  Let $U, V$ and $\prag{P}$ be as above. The following three statements are equivalent:
\begin{enumerate} \item  $\prag{P}$ is calibrated
  for ${U} \mid_{\we} {V}$; 
\item There exists a RV $V'$ on $\cZ$ with $V \Pdetermines{\prag{P}} V'$ such that
$\prag{P}$ is safe for $U \mid \cali{V} ,V'$
%
\item $\prag{P}$ is safe for ${U} \mid {V''}$ where
$V''$ is the generalized RV given by $V'' \equiv \prag{P}(U \mid V)$.
\end{enumerate}
\end{theorem}
Note that, since safety for $U \mid V$ implies safety for $U \mid
\cali{V}, V'$ for $V'= V$, {\em (2.) $\Rightarrow$ (1.)\/} shows that
safety for $U \mid V$ implies calibration for $U \mid V$. By mere
definition chasing (details omitted) one also finds that {\em (2.)\/}
implies that $\tilde{P}$ is safe for $U \mid \ave{V},V'$ and, again by
definition chasing, that $\tilde{P}$ is safe for $U \mid
\ave{V}$. Thus, this result establishes two more arrows of the
hierarchy of Figure~\ref{fig:figure}.  Its proof is based on the
following simple result, interesting in its own right:
\begin{proposition} \label{prop:ignore} Let $V$ and $V'$ be
  generalized RVs such that $V \fPdetermines{f}{\prag{P}} V'$ for some function $f$.  The following statements are
  equivalent:
\begin{enumerate} 
\item $\prag{P}(U \mid V,V')$ ignores $V$, i.e. $\prag{P}(U \mid V,V') =_{\prag{P}} \prag{P}(U \mid V')$. 
\item For all $v' \in \support_{\prag{P}}(V')$, 
for all $v \in \support_{\prag{P}}(V)$ with $f(v) = v'$: $\prag{P}(U \mid V =
v) = P(U \mid V'= v')$.
\item $V' \Pdetermines{\prag{P}} \prag{P}(U \mid V)$
\item $V' \Pdetermines{\prag{P}} V''$ and $\prag{P}(U \mid V',V'')$
  ignores $V'$, where $V'' = \prag{P}(U \mid V)$.
\end{enumerate} Moreover, if $\prag{P}$
is safe for $U \mid V$ and $\prag{P}(U \mid V, V')$ ignores $V$, then $\prag{P}$ is safe for $U \mid V'$.
\end{proposition}

\section{Continuous-Valued $U$ and $V$; Confidence and Pivotal Safety}
\label{sec:continuous}
Our definitions of safety were given for countable $\cZ$, making all
random variables involved have countable range as well.  Now we allow
general $\cZ$ and hence continuous-valued $U$ and general uncountable
$V$ as well, but we consider a version of safety in which we do not
have safety for $U \mid V$ itself, but for $U'\mid V$ for some $U'$
with $(U,V) \determines U'$ such that the range of
$\prag{\cP}_{[V]}(U') := \{ \prag{P}(U' \mid V=v): v \in \range{V} \}$
is still countable. To make this work we have to equip $\cZ$ with an
appropriate $\sigma$-algebra $\Sigma_{\cZ}$ and have to add to the
definition of a RV that it must be measurable,\footnote{Formally we
  assume that $\cZ$ is equipped with some $\sigma$-algebra
  $\Sigma_{\cZ}$ that contains all singleton subsets of $\cZ$.  We
  associate the co-domain of any function $X: \cZ \rightarrow \reals^k$
  with the standard Borel $\sigma$-algebra on $\reals^k$, and we call
  such $X$ an RV whenever the $\sigma$-algebra $\Sigma_{\cZ}$ on $\cZ$
  is such that the function is measurable.}  and we have to modify the
definition of support $\support_{P}(U)$ to the standard
measure-theoretic definition (which specializes to our definition
(\ref{eq:range}) whenever there exists a countable $\cU$ such that
$P(U \in \cU) = 1$). Yet nothing else changes and all previous
definitions and propositions can still be used.\footnote{If we were to
  consider safety of the form $U \mid V$ for uncountable
  $\prag{\cP}_{[V]}(U)$, then this set of probability distributions
  would have to be equipped with a topology, which is a bit more
  complicated and is left for future work.}

\paragraph{Additional Notations and Assumptions}
In this section we frequently refer to (cumulative) distribution
functions of 1-dimensional RVs, for which we introduce some notation:
for distribution $P \in \cP^*$ and RV $W: \cZ \rightarrow \reals$, let
${F}_{[W]}(\cdot)$ denote the distribution function of $W$, i.e.
${F}_{[W]}(w) = P(W \leq w)$. The notation is extended to
conditional distribution functions: for given $\prag{P}(U \mid V)$, we let
$\cdf(u |v) := \prag{P}(U\leq u \mid V=v )$.  The subscripts $[W]$ and
$[U|V]$ indicate the RVs under consideration; we will omit them if
they are clear from the context.  Note that we can consider these
distribution functions as RVs: for all $z \in \cZ$,
${F}_{[W]}(W)(z) = {P}(W \leq W(z))$ and
$\cdf(U |V)(z) = \prag{P}(U \leq U(z) \mid V= V(z))$.

Since this greatly simplifies matters, we will often assume that
either $P \in \cP^*$ or
$\tilde{P}(U \mid V)$ satisfy the following:
\paragraph{Scalar Density Assumption} 
A distribution $P(U)$ for RV $U$ satisfies the {\em scalar density
  assumption\/} if (a) $\range{U}\subseteq \reals$ is equal to some
(possibly unbounded) interval, and (b) $P$ has a density $f$
relative to Lebesgue measure with $f(u) > 0$ for all $u$ in the interior of $\range{U}$. We say that $P(U|V)$ satisfies the scalar density assumption if for all $v \in \range{V}$, $P(U \mid V=v)$ satisfies it. \\ \ \\
This is a strong assumption which will nevertheless be satisfied in
many practical cases. For example, normal distributions, gamma
distributions with fixed shape parameter, beta distributions etc. all
satisfy it.
\paragraph{Overview of this Section}
The goal of the following two subsections is to precisely reformulate the {\em
  fiducial\/} and {\em confidence\/} distributions that have been
proposed in the statistical literature as pragmatic distributions in
our sense, that can be safely used for some (`confidence-related') but
not for other prediction tasks.  Here we focus on the standard statistical
scenario introduced in Example~\ref{ex:confidence}.  
The underlying idea of `pivotal safety' (developed in
Section~\ref{sec:pivots}) has applications in discrete, nonstatistical settings as well, as explored in Section~\ref{sec:decision}.  

\subsection{Confidence Safety}
\label{sec:confidence}
We start with a classic motivating example.
\begin{example}{\bf \ [Example~\ref{ex:confidence},
    Specialized]} \label{ex:normal} As a special case of the
  statistical scenario outlined in Example~\ref{ex:confidence}, let
  $\cM$ be the normal location family with varying mean $\theta$ and
  fixed variance, say $\sigma^2 = 1$, and let
  $V := \hat\theta = \hat{\theta}(X^n)$ where $\hat{\theta}(X^n)$ is
  the empirical average of the $X_i$, which is a sufficient statistic
  that is of course also equal to the ML estimator for data
  $X^n$. Then the sampling density of $\hat\theta$ is itself Gaussian,
  and given by
\begin{equation}\label{eq:ziek}
p(\hat\theta \mid \theta) \propto q_{\theta}(X^n ) \propto e^{- \frac{1}{2}\cdot n \cdot {(\hat\theta- \theta)^2}}.
\end{equation}
In this simple context, Fisher's controversial fiducial reasoning
amounts to observing that (\ref{eq:ziek}) is symmetric in $\hat\theta$ and
$\theta$; thus, if we simply define a new function
$\tilde{p}(\theta \mid \hat\theta) := p(\hat\theta \mid \theta)$, then this
function must, for each fixed $\hat\theta$, be the density of a
probability distribution (the integral over $\theta$ must by symmetry be
1); and this would then amount to something like a `prior--free'
posterior for $\theta$ based on data $\hat\theta$. In this special case, as
well as with the corresponding inversion for scale families, it
coincides with the Bayes posterior based on an improper Jeffreys'
prior. Yet, \cite{lindley1958fiducial} showed that the general
construction for 1-dimensional families, which we review in the next
subsection, {\em cannot\/} correspond to a Bayesian posterior except
for location and scale families: for different sample sizes, the
`fiducial' posterior for data of size $n$ corresponds to the Bayes
posterior for a prior which depends on $n$.
\end{example}
\cite{Fisher30} noted that $\tilde{p}$ as constructed above lead to
valid inference about confidence intervals. Later \citep{Fisher35} he
made claims that $\tilde{p}$ could be used for general prior-free
inference about $\theta$ given data/statistic $\hat\theta$. This is not
correct though, and more recently, $\tilde{p}$ is more often regarded
as an instance of a {\em confidence distribution\/}
\citep{SchwederH02}, a term going back to \cite{Cox58} --- these are
by and large the same objects as fiducial distributions, though with a
stipulation that they only be used for certain inferences related to
confidence. In the remainder of this subsection, we develop a
variation of safety that can capture such confidence statements. In
the next subsection, we review the general method for designing
confidence distributions for 1-dimensional statistical families and we
shall see that, under an additional condition, they are indeed
confidence--safe in our sense. 
In the remainder of this section, we focus on 1-dimensional families
and interpret the RVs $U$ and $V$ as in our statistical application of
Example~\ref{ex:confidence} and~\ref{ex:normal}. Thus, $U \equiv \theta$ would be a
1-dimensional scalar parameter of some model $\{P_{\theta} : \theta \in \Theta \}$, $V \equiv S(X^n)$ would be a statistic
of the observed data. In Example~\ref{ex:normal},
$V \equiv \hat\theta(X^n)$ is the ML estimator.

We are thus interested in constructing, for each $v \in \range{V}$, an
interval of $\reals$ that has (say) 95\% probability under
$\prag{P}(U\mid V=v)$.  To this end, we define for each $v \in \range{V}$, an
interval $\cs_v = [\underline{u}_v,\bar{u}_v]$ where $\underline{u}_v$
is
such that $\cdf(\underline{u}_v \mid v) = 0.025$ 
and $\bar{u}_v$ is 
such that $\cdf(\bar{u}_v \mid v) = 0.975$. This set
obviously has $95\%$ probability according to $\prag{P}(U \mid
V=v)$. In our interpretation where $U = \theta$ is the parameter of a
statistical model, we may interpret $\prag{P}$ as DM's assessment,
given data $V = S(X^n)$, of the uncertainty about $U$,
i.e. $\prag{P}$ is a `posterior' and, analogous to Bayesian
terminology, we may call $\cs_v$ a $95\%$ {\em credible set\/} given
$V$. The question is now under what conditions we have {\em coverage},
i.e. that $\cs_V$ is also a $95\%$ frequentist {\em confidence
  interval}, so that our credible set can be given frequentist
meaning. By definition of confidence interval, this will be the case
iff for all $P \in \cP^*$, $ P(U \in \cs_V) = 0.95, $ i.e. iff for all
$P \in \cP^*$, $v \in \range{V}$,
\begin{equation}\label{eq:mainconfidencea}
\Exp_{P}[\indicator_{U \in \cs_V}] = \Exp_{\prag{P}} [\indicator_{U \in \cs_V} \mid V=v],
\end{equation}
where we used that, by construction, $\Exp_{\prag{P}} [\indicator_{U
  \in \cs_V} \mid V=v]=0.95$ for all $v \in \range{V}$. 
As we shall see (\ref{eq:mainconfidencea}) holds for our
normal example, so the posterior constructed in (\ref{eq:ziek})
produces valid confidence intervals. (\ref{eq:mainconfidencea}) is of
the form of a `safety' statement and it suggests that confidence
interval validity of credible sets can be phrased in terms of safety
in general. Indeed this is possible as long as $\tilde{P}(U|V)$ satisfies the scalar density assumption: for fixed $0 \leq a < b \leq 1$,
we can define the  set $\ct_v = [\underline{u}^a_v,\bar{u}^b_v]$ where
$\cdf(\underline{u}_v^a \mid v)
= a$ and 
$\cdf(\bar{u}^b_v \mid v) = b$, so that for each $v \in \range{V}$, $\ct_v$ is a $b-a$
credible set. Reasoning like above, we then get that $\ct_V$ is also a
$b-a$ confidence interval iff for all $P \in \cP^*$, all $v \in \range{V}$
\begin{equation}\label{eq:mainconfidenceb}
\Exp_{(U,V) \sim P}[\indicator_{U \in \ct_V}] = \Exp_{P} \Exp_{\prag{P}} 
[\indicator_{U \in \ct_V} \mid V=v ],
\end{equation}
which, from the characterization of safety for $U \mid [V]$,
Proposition~\ref{prop:newstart}, (\ref{eq:calibrationagain}) and
(\ref{eq:validityrevisited}) suggests the following definition:
\begin{definition}\label{def:confidence}
  Let $U$, $V$ and  $\prag{P}$ be such that $\tilde{P}(U|V=v)$ satisfies the scalar density assumption for all $v \in \range{V}$. We say that
  $\prag{P}$ is (strongly) {\em confidence--safe\/} for $U \mid V$ if for all
  $0 \leq a < b \leq 1$, it is safe for
  $\indicator_{U \in \ct_V} \mid \cali{V}$.
\end{definition}
The requirement that $\tilde{P}$ satisfies the scalar density
assumption is imposed because otherwise $\ct_V$ may not be defined for some $a,b \in [0,1]$. 
We could also consider distributions
that have coverage in a slightly weaker sense, and define weak
confidence-safety for $U \mid V$ as safety for
$\indicator_{U \in \ct_V} \mid \ave{V}$; we have not (yet) found any
natural examples though that exhibit weak confidence-safety but not
strong confidence safety. 
\begin{example}\label{ex:supernormal}
  In the next subsection we show that $\prag{P}(\theta \mid V =
  \hat\theta(X^n))$ as defined in Example~\ref{ex:normal} (normal
  distributions) is confidence-safe. For example, we may specify a
  $\tilde{P}(\theta \mid \hat\theta)$--95\% credible set $\ct_V$ with
  $a=0.025$ and $b = 0.975$ as the area under the normal curve
  centered at $V= \hat\theta$ and truncated  so that the area under
  the left and right remaining tails is $0.025$ each. Now suppose that
  $X^n \sim P_{\theta}$ for arbitrary $\theta$. By confidence--safety
  we know that the probability that we will observe $\hat\theta$ such
  that $\theta \not \in \ct_V$ is exactly $0.05$, just as it would be
  if $\tilde{P}$ where the true conditional distribution --- an instance of a {\em safe\/} inference based on $\tilde{P}$. For an
  example of an inference that is {\em unsafe}, suppose DM really is
  offered a gamble for \$1 that pays out \$2 whenever $\theta > 0$ (we
  could take any other fixed value as well), and pays out $0$
  otherwise. She thus has two actions at her disposal, $a=1$ (accept
  the gamble) and $a=0$ (abstain), with loss given by $L(\theta,0) =
  0$ for all $\theta$ and $L(\theta,1) = 1$ if $\theta < 0$ and
  $L(\theta,1) = -1$ otherwise. She might thus be tempted to follow
  the decision rule $\delta(\hat\theta)$ that accepts the gamble
  whenever she observes $\hat\theta$ such that $\tilde{P}(\theta > 0
  \mid \hat\theta) > .5$ and abstains otherwise; for that rule
  minimizes, among all decision rules, her expected loss $\Exp_{\theta
    \sim \tilde{P} \mid \hat\theta}[L(\theta,\delta(\hat\theta))]$,
  and gives negative expected loss.

  This decision rule should not be followed though, because it is
  based on an inference that is not safe in any of our senses: safety
  would mean that $\tilde{P}$ is safe for
  $L(\theta,\delta(\hat\theta)) \mid \text{\sc s}$, where $\text{\sc s}$ can be
  substituted by $\cali{\hat\theta}$, $\ave{\hat\theta}, $ or
  ${\hat\theta}$. The first does not apply since $\hat\theta$ is not
  ignored in the probability assessment; the third does not hold because it would imply the second, which also does not hold. To see this,
note that if data comes from $P_{\bar\theta}$ with $\bar\theta < 0$ then we have
$$
\Exp_{\hat\theta \sim P_{\bar\theta}} [L(\bar\theta,\delta(\hat\theta))] > 0 > \Exp_{\hat\theta \sim P_{\bar\theta}} [\Exp_{\theta \sim \tilde{P} \mid
    \hat\theta}[L(\theta,\delta(\hat\theta))]], 
$$
so that her actual expected loss is positive whereas she thinks it to
be negative. This violates (\ref{eq:startsafea}) in
Definition~\ref{def:leftsafety} so that $\tilde{P}$ is not safe for
$L(\theta,\delta(\hat\theta)) \mid \ave{\hat\theta}$. Note that, if
$\tilde{P}$ were safe for $\theta \mid \hat\theta$ (as a subjective
Bayesian would believe if $\tilde{P}$ were her posterior) then it
would also be safe for $L(\theta,\delta(\hat\theta))$ (because
$L(\theta,\delta(\hat\theta))$ can be written as a function of
$(\theta,\hat\theta)$), and then use of $\delta$ would be safe after
all.

For an intuitive interpretation, consider a long sequence of
experiments. For each $j$, in the $j$-th experiment, a sample of size
$n=10$ is drawn from a normal with some mean $\theta_j$. Each time DM
investigates whether $\theta_j > 0$. Assume that, in reality, all of
the $\theta_j$ are $< 0$, but DM does not know this. Then every once
in a while $\hat\theta$ will be large enough for our unsafe DM to
gamble on it, but every time this happens she loses; all other times
she neither loses nor wins, so her net gain is negative in the long
run.

Thus, $\tilde{P}$ is not safe for $\theta \mid \hat\theta$ in
general. However, it is still safe for $U' |V$ for some other
functions of $U \equiv \theta$ besides $U' = \ct_V$. For example, it
leads to unbiased estimation of the mean: $\tilde{P}$ is safe for
$\ave{\theta} \mid \ave{\hat\theta}$, as is easily established. This
is however a special property of the confidence distribution for the
normal location family and does not hold for general 1-dimensional
confidence distributions as reviewed below.
\end{example}
\subsection{Pivotal Safety and Confidence}
\label{sec:pivots}
Trivially, if $\tilde{P}$ is safe for $U|V$ (hence valid) and the
scalar density assumption holds, then it is also confidence-safe for
$U|V$.  We now determine a way to construct confidence-safe
$\tilde{P}$ if not enough knowledge is available to infer a
$\tilde{P}$ that is valid.  To this end, we invoke the concept of a
{\em pivot}, usually defined as a function of the data and the
parameter that has the same distribution for every $P \in \cP^*$ and
that is monotonic in the parameter for every possible value of the
data \citep{BarndorffNielsenCox94}. We adopt the following variation
that also covers a quite different situation with discrete outcomes:
\begin{definition}{\bf \ [pivot]} \label{def:pivot} Let $U$ and $V$ be
  as before and suppose either (continuous case) that
  $U: \cZ \rightarrow \reals$ and $V: \cZ \rightarrow \reals$ are
  real-valued RVs, and that for all $v\in \range{V}$, $\range{U|V=v}$
  is a (possibly unbounded) interval (possibly dependent on $v$), or
  (discrete case) that $\cZ$ is countable.  We call RV $U'$ a
  (continuous viz. discrete) {\em pivot\/} for $U \mid V$ if
\begin{enumerate}
\item $(U,V) \determines U'$ so that the function $f$ with $U'= f(U,V)$ exists.
\item For each fixed $v \in \range{V}$, the function $f_v:
  \range{U|V=v} \rightarrow \range{U'}$, defined as $f_v(u) := f(u,v)$
  is  1-to-1 (an injection); in the continuous case we further require
  $f_v$ to be continuous and uniformly monotonic, i.e. either $f_v(u) $ is increasing
  in $u$ for all $v \in \range{V}$, or $f_v(u)$ is decreasing in $u$
  for all $v \in \range{V}$.
\item All $P \in \cP^*$ agree on $U'$, i.e. for all $P_1, P_2 \in
  \cP^*$, $P_1(U') = P_2(U')$, where in the continuous case we further
  require that $P_1$ (hence also $P_2$) satisfies the scalar density assumption. 
\end{enumerate}
We say that a pivot $U'$ is {\em simple\/} if for all $v \in \range{V}$, 
the function $f_v$ is a bijection. 
\end{definition}
The scalar density assumption (item 3) does not belong to the standard
definition of pivot, but it is often assumed implicitly, e.g. by
\cite{SchwederH02}. The importance of `simple' pivots (a nonstandard
notion) will become clear below.

In the remainder of this section we focus on the statistical case of
the previous subsection, which is a special case of
Definition~\ref{def:pivot} above --- thus invariably $U \equiv
\theta$, the 1-dimensional parameter of a model $\{P_{\theta} \mid
\theta \in \Theta \}$, and $V$ is some statistic of data $X^n$. In
Section~\ref{sec:decision} we return to the discrete case.

If a continuous pivot as above exists, then all $P \in \cP^*$ have the
same distribution function $F_{[U']}(u') := P(U'\leq u')$. We may thus
define a pragmatic distribution by setting, for all $v \in \range{V}$,
all $u \in \range{U \mid V=v}$,
\begin{equation}\label{eq:belangrijker}
\tilde{F}_{[U|V]}(u
\mid v) := \begin{cases} F_{[U']}(f_v(u)) & \text{if $f_v(u)$ increasing in $u$} \\
1 - F_{[U']}(f_v(u)) &  \text{if $f_v(u)$ decreasing.}
\end{cases}
\end{equation} 
The definition of pivot ensures that for each $v \in \range{V}$,
$\tilde{F}_{[U|V]}(u \mid v)$ is a continuous increasing function of
$u$ that is in between $0$ and $1$ on all $u \in \range{U \mid V=v}$,
and hence $\tilde{F}_{[U|V]}(u \mid v)$ is the CDF of some
distribution $\tilde{P}(U|V)$.  It can be seen from the standard
definition of a confidence distribution \citep{SchwederH02} that this
$\prag{P}(U|V)$ is a confidence distribution, and that every
confidence distribution can be obtained in this
way.\footnote{Mirroring the discussion underneath Definition 1 from
  \cite{SchwederH02}: if $\tilde{F}'(U|V)$ is the CDF of a confidence
  distribution as defined by them, then $U':= \tilde{F}'(U|V)$ is a pivot and then the
  construction above applied to $U'$ gives
  $\tilde{F}(U|V) := \tilde{F}'(U|V)$. Conversely, if $U'$ is an
  arbitrary continuous pivot, then by the requirement that $P(U')$ has
  a density with interval support, $F(U')$ is itself uniformly
  distributed on its support $[0,1]$ and there is a 1-to-1 continuous
  mapping between $U'$ and $F(U')$. Thus, whenever $U'$ is a
  continuous pivot, $F(U')$ is itself a pivot as well, and
  $\tilde{F}(U|V)$ as defined here satisfies the definition of
  confidence distribution.} Hence, (\ref{eq:belangrijker}) essentially
defines confidence distribution. Theorem~\ref{thm:confidence} below
shows that when based on {\em simple\/} pivots, confidence distributions are
also confidence-safe.
\begin{example}\label{ex:prefreund}
  Consider  the statistical setting  with $U
  \equiv \theta$, $V \equiv \hat\theta(X^n)$, and (a) 
for all $\theta \in \Theta$,
  $P_{\theta}(V)$ itself satisfies the scalar density assumption, and (b)
  for each fixed $v \in \range{V}$, we have that $F_{\theta,[V]}(v) :=
  P_{\theta}(\hat\theta(X^n) \leq v)$ is monotonically decreasing in
  $\theta$. This will hold for 1-dimensional exponential
  families with a continuously supported sufficient statistic (such as
  the normal, exponential, beta- and many other models), taken in
  their mean-value parameterization $\Theta$.  Then (by (b)) $U' = F_{\theta,[V]}(V)$ is itself a decreasing pivot, with (by (a)) the additional property that the function $f_{\theta}$ from $\range{V}$ to $\range{U'}$ given by $f_{\theta}(v) := f(\theta,v)$ is strictly increasing in $v$. 
Then (\ref{eq:belangrijker}) simplifies, because (using this strict increasingness in the second equality):
$$F_{\theta,[V]}(v) = P_{\theta}(V \leq v) = P_{\theta}(F_{\theta,[V]}(V) \leq f_{\theta}(v)) = F_{\theta, [F_{\theta,[V]}]}(f_{\theta}(v)) = F_{[U']}(f_v(\theta)),$$
and noticing that the right-hand side appears in (\ref{eq:belangrijker}), we can plug in the left-hand side there as well and we see that we can directly set
\begin{equation}\label{eq:belangrijkst}
\tilde{F}(\theta \mid \hat\theta) = 1 - F_{\theta}(\hat\theta).
\end{equation}
Thus for such models the recipe (\ref{eq:belangrijker}) simplifies (see also \cite{veronese2015fiducial}).
\end{example}
We now define `pivotal safety' which, as demonstrated below, in the
statistical case essentially coincides with confidence safety --- the
reason for the added generality is that it also has meaning and
repercussions in the discrete case.  The extension to `multipivots' is
just a stratification that means that, given any $w \in \range{W}$,
$\tilde{P} \mid W=w$ is pivotally safe for $U \mid V$ relative to
$\cP^* \mid W=w$; it is not really needed in this text, but is
convenient for completing the hierarchy in Figure~\ref{fig:figure}.

\begin{definition}\label{def:pivotal} Let $U$ and $V$ be as before and let 
  $\tilde{P}$ be an arbitrary distribution on $\cZ$ (not necessarily
  given by (\ref{eq:belangrijker})) .  If $V$ has full support under
  $\prag{P}$, i.e. $\support_{\prag{P}}(V) = \range{V}$ and $U'$ is a
  (continuous or discrete) pivot such that $\prag{P}$ is safe for $U'
  \mid \cali{V}$, i.e. for all $v \in \range{V}$, 
$$
\tilde{P}(U'\mid V=v) = \tilde{P}(U'),
$$
then we say that $\tilde{P}$ is pivotally safe for $U \mid V$, with
pivot $U'$.

Now let $W$ be a generalized RV such that $V \determines W$.  Suppose that for all $w \in
\range{W}$, $U'$ is a pivot relative to the set of distributions $(\cP^*
\mid W =w)$ and $\tilde{P}$ is safe for $U' \mid \cali{V},
W$. Then we say that $\tilde{P}$ is pivotally safe for $U \mid V$
with {\em multipivot\/} $U |W$.
\end{definition} 

\begin{example}{\bf [normal distributions and general confidence distributions]}\label{ex:freund}
  In Example~\ref{ex:normal}, $U = \theta$, the mean of a normal with
  variance $1$, and we set $V = \hat\theta(X^n)$ to be the average of
  a sample of size $n$. Then it is easily seen that $U' = \theta -
  \hat\theta = U-V$ is a pivot according to our definition, having a
  $N(0,1)$ distribution under all $P \in \cP^*$.  If we adopt
  $\prag{P}(U \mid V)$, under which $U' \sim N(0,1)$ independently of
  $V$, then $\tilde{P}$ is safe for $U' \mid [V]$ (see
  Proposition~\ref{prop:newstart}, (\ref{eq:calibrationagain})) so we
  have pivotal safety. And indeed one can verify that this $\tilde{P}$
  coincides with the recipe given by (\ref{eq:belangrijker}).
\end{example}

While in the simplest form of calibration, safety for $U \mid
\cali{V}$, we had that $\tilde{P}(U \mid V)$ was the marginal of $U$,
so that $U$ and $V$ are independent under $\tilde{P}$, in the simplest
pivotal safety case, the situation is comparable, but now the
auxiliary variable $U'$ instead of the original variable $U$ is
independent of $V$ under $\prag{P}$.  Note though that we do not
necessarily have that $U'$ and $V$ are independent for all or even for
{\em any\/} $P \in \cP^*$. In Example~\ref{ex:montyhallb} (Monty Hall)
below, there is in fact just one single $P \in \cP^*$ for which $U'
\perp V$ holds. In the statistical application, we even have that
$P(U' \mid V)$ is a degenerate distribution (putting all its mass on a
particular real number depending on $V$) under each $P \in \cP^*$, as
can be checked from Example~\ref{ex:normal}.

The relation between pivotal, confidence-safety and the $\tilde{P}$
defined as in (\ref{eq:belangrijker}) is given by the following
central result.

\begin{theorem}\label{thm:confidence}
The following statements are all equivalent:
\begin{enumerate}
\item $\tilde{P}$ is pivotally safe for $U|V$ with some {\em simple\/}
  continuous pivot $U'$
\item $\prag{P}(U|V)$ is of form (\ref{eq:belangrijker}), where $U'$
  is a simple continuous pivot
\item $\prag{P}$ is confidence-safe for $U|V$  and for each $v \in \range{V}$, $\prag{P}(U|V=v)$ satisfies the scalar density assumption
\item $\prag{P}$ is safe for $\prag{F}(U|V) |V $ and for each $v \in \range{V}$, $\prag{P}(U|V=v)$ satisfies the scalar density assumption
\item $\prag{P}$ is pivotally safe for $U|V$ with pivot $\prag{F}(U|V)$, which is continuous and simple.
\end{enumerate}\end{theorem}
The theorem shows that, whenever pivots are simple (as is often the
case), confidence distributions $\tilde{P}$ as defined by
(\ref{eq:belangrijker}) are also confidence-safe. If a pivot is
nonsimple however, confidence distributions can still be defined via
(\ref{eq:belangrijker}) but they may not be confidence-safe under our
current definition. An example of such a case is given by the
statistical scenario where $\cM$ is the 1-dimensional normal family,
but the parameter of interest is $U := |\theta|$ rather than
$\theta$. As shown by \cite{SchwederH16}, the confidence distribution
$\tilde{P}(U\mid \hat\theta)$ defined by (\ref{eq:belangrijker}) then
gives a point mass $\tilde{P}(U=0 \mid \hat\theta = v) = p > 0$ to
$U=|\theta| = 0$ whose size $p$ depends on $v$. This happens because
$\tilde{F}(U \mid v )$ ranges, for such $v$, not from $0$ to $1$ but
from some $a> 0$ (depending on $v$) to $1$. Then pivotal safety cannot
be achieved for $\tilde{P}(U|V)$, since there must be $v_1, v_2$ with
$\tilde{P}(U |V=v_1) \neq \tilde{P}(U |V=v_2)$ whence the definition
is not satisfied.  Now such confidence distributions based on
nonsimple pivots are still useful, and indeed we can prove a weaker
form of pivotal and confidence safety for such cases, by replacing
safety for $U' \mid \cali{V}$ in Definition~\ref{def:pivotal} by
safety for $U' \mid \wksafe{V}$; we will not discuss details here
however.

\paragraph{The Hierarchy}
To see how pivotal and confidence safety fit into the hierarchy of
Figure~\ref{fig:figure}, note that Theorem~\ref{thm:confidence}
establishes the double arrow between pivotal safety and confidence
safety under the scalar density assumption (SDA) --- the requirement
that $(U',V) \determines U$ in the figure amounts to $f_v$ being a
bijection, as we require. The theorem also establishes the relation
between calibration and pivotal safety, under the assumption that
$\tilde{P}(V)$ has full support and the SDA holds for $U$. Then the
simplest form of calibration, safety for $U \mid \cali{V}$, clearly
implies pivotal safety for $U \mid V$ --- just take $U'= U$, which is
immediately checked to be a pivot. This result trivially extends to
the general case of safety for $U \mid \cali{V},V'$ with $V'\neq {\bf
  0}$, this implying pivotal safety with multipivot $U \mid V'$ --- we
omit the details.

It remains to establish the rightmost column of
Figure~\ref{fig:figure}; we will only do this in an informal
manner. \cite{SchwederH02} (and, implicitly, \cite{Hampel06}) already
note that if $\tilde{P}$ is a confidence distribution for RV $U$ given
data $V$, then it remains a confidence distribution for monotonic
functions $U'$ of $U$, but not for general functions of $U$.  In our
framework this translates to, under the scalar density assumption of
Section~\ref{sec:confidence}, that pivotal safety of $U \mid V$
implies pivotal safety for $U' \mid V$ if $U'$ is a 1-to-1 continuous
function of $U$, which readily follows from Definition~\ref{def:pivot}
and Theorem~\ref{thm:confidence} (Definition~\ref{def:pivot} implies
an analogous statement for the discrete case as well). Similarly, it
is a straightforward consequence from the definitions that calibration
for $U \mid V$ implies calibration for $U' \mid V$, for every $U'$
with $U \determines U'$, not necessarily 1-to-1; yet for $U'$ with
$(U,V) \determines U'$, calibration may not be preserved: take
e.g. the setting of Example~\ref{ex:dilation} (dilation) with $U'=
|V-U|$. Then $\tilde{P}(U'= 1 \mid V= 0) = 0.9$, $\tilde{P}(U'= 1 \mid
V=1) = 0.1$, yet $\cP^*$ contains a distribution with $P(U=V) = 1$ and
for this $P$, $P(U'=1 \mid V) \equiv 0$. If $\tilde{P}$ is valid for
$U \mid V$ however, validity is preserved even for every $U'$ with
$(U,V) \determines U'$.
\subsection{Pivotal Safety and Decisions}
\label{sec:decision}
Now we consider pivotal safety for RVs $U$ with countable
$\range{U}$. The Monty Hall example below shows that in this case,
pivotal safety is still a meaningful concept.  We first provide an
analogue of Theorem~\ref{thm:confidence}, in which `confidence safety'
is replaced by something that one might call `local' confidence
safety: safety for a RV $U'$ that determines the {\em probability of
  the actually realized outcome\/} $U$. To this end, we introduce some
notation: for distribution $P \in \cP^*$ and RV $W$, let
$\pmf{W}(\cdot)$ be the RV that denotes the probability mass function
of $W$, i.e.  $\pmf{W}(w) = P(W = w)$; similarly $\prag{p}_{[W]}(w) :=
\prag{P}(W = w)$. The notation is extended to conditional mass
functions as $\tpmf{U|V}(u |v) := \prag{P}(U= u \mid V=v )$.  The
subscript $[W]$ and $[U|V]$ indicates the RVs under consideration; we
will omit them if they are clear from the context.  We can think of
these mass functions as RVs: for all $z \in \cZ$, $p_{[W]}(W)(z) = P(W
= W(z))$;
$\prag{p}_{[U|V]}(U |V)(z) = \prag{P}(U =
U(z) \mid V= V(z))$. The difference between RV $\prag{P}(U \mid V)$
and $\prag{p}_{[U|V]}(U |V)$ is that the former maps $z$ to the {\em
  distribution\/} $\prag{P}(U \mid V= V(z))$; the latter maps $z$ to
the {\em probability of a single outcome\/} $\prag{P}(U = U(z) \mid V= V(z))$.
\begin{theorem}\label{thm:discreteconfidence}
  Let $\cZ$ be countable, $U$ be an RV and $V$ a generalized RV.
  Suppose that for all $v \in \range{V}$, all $p \in
  [0,1]$, $\#\{u: \prag{P}( U=u \mid V=v) = p \} \leq 1$ (i.e. there
    are no two outcomes to which $\prag{P}(U \mid V=v)$ assigns the
    same nonzero probability).  Then the following statements are all
    equivalent:
\begin{enumerate}
\item $\prag{P}$ is safe for $\prag{p}(U|V) \mid [V]$.
\item $\prag{P}$ is pivotally safe for $U \mid V$, with simple pivot
$U'= \prag{p}(U|V)$.
\item $\prag{P}$ is pivotally safe for $U \mid V$ for some simple pivot $U''$.
\end{enumerate}
\end{theorem}
This result establishes that, in the discrete case,  if there is {\em some\/}
simple pivot $U''$, then  $\prag{p}(U|V)$ is also a simple  pivot --- thus $\prag{p}(U|V)$ has some generic status. Compare this to 
Theorem~\ref{thm:confidence} which established that $\tilde{F}(U|V)$ is a generic pivot in the continuous case.

We now illustrate this result, showing that, for a wide range of loss
functions, pivotal safety implies that DM has the right idea of how
good her action will be if she bases her action on the belief that
$\tilde{P}$ is true --- even if $\tilde{P}$ is false.  Consider an RV
$U$ with countable range $\cU := \range{U}$. Without loss of
generality let $\cU = \{1, \ldots, k \}$ for some $k > 0$ or $\cU =
\naturals$. Let $L: \cU \times \cA \rightarrow \reals \cup \{
\infty\}$ be a {\em loss function\/} that maps outcome $u \in \cU$ and
action or decision $a \in \cA$ to associated loss $L(U,a)$. We will
assume that $\cA \subset \Delta(\cU)$ is isomorphic to a subset of the
set of probability mass functions on $\cU$, thus an action $a$ can be
represented by its (possibly infinite-dimensional) mass vector $a=
(a(1),a(2), \ldots)$. Thus, $L$ could be any scoring rule as
considered in Bayesian statistics (then $\cA = \Delta(\cU)$), but it
could also be $0/1$-loss, where $\cA$ is the set of point masses
(vectors with $1$ in a single component) on $\cU$, and $L(u,a)= 0$ if
$a(u) = 1$, $L(u,a) = 1$ otherwise.  For any bijection $f: \cU
\rightarrow \cU$ we define its extension $f: \cA \rightarrow \cA$ on
$\cA$ such that we have, for all $u \in \cU$, all $a \in \cA$, with
$a'= f(a)$ and $u'= f(u)$, $a'(u') = a(u)$. Thus any $f$ applied to
$u$ permutes this outcome to another $u'$, and $f$ applied to a
probability vector permutes the vector entries accordingly.

We say that $L$ is {\em symmetric\/} if for all bijections $f$,
all $u \in \cU, a \in \cA$, $L(u,a) = L(f(u),f(a))$. This requirement says that the loss is invariant under
any permutation of the outcome and associated permutation of the
action; this holds for important loss functions such as the
logarithmic and Brier score and the randomized $0/1$-loss, and many others.

We will also require that for all distributions $P$ for  $U$, there
exists at least one Bayes action $a_{P} \in \cA$ with $E_{P}[L(U,a_{P})] =
\min_{a \in \cA} E_{P}[L(U,a)]$ --- which again holds for the
aforementioned loss functions.  If there is more than one such act we
take $a_P$ to be some arbitrary but fixed function that maps each $P$
to associated Bayes act $a$.  In the theorem below we abbreviate
$a_{\prag{P}(U \mid V)}$ (the optimal action according to $\tilde{P}$
given $V$, i.e. a generalized RV that is a function of $V$) to
$\prag{a}_{V}$.
\begin{theorem}\label{thm:loss}
  Let $\prag{P}(U \mid V)$ be a pragmatic distribution where
  $\cZ$ is countable. Suppose that $\prag{P}(U \mid V)$ is
  pivotally safe with a simple pivot. Let $L: \cU \times \cA \rightarrow \reals \cup \{
  \infty\}$ be a symmetric loss function as above, and let
  $\prag{a}_V$ be defined as above. Then $\prag{P}$ is safe for
  $L(U,\prag{a}_V) \mid [V]$, i.e. for all $v \in \range{V}$, all $P
  \in \cP^*$,
$$
\Exp_{(U,V) \sim P}[L(U,\prag{a}_{V})] = \Exp_{\prag{P}}[L(U,\prag{a}_{V}) 
\mid V=v].
$$
\end{theorem}

\begin{example}\label{ex:montyhallb}{\bf [Use of Pivots beyond Statistical Inference: Monty Hall]} 
To illustrate Theorem~\ref{thm:loss},  consider again the Monty Hall problem (Example~\ref{ex:montyhall}) where the contestant chooses
  door $1$. We model this using RV $U \in \{1,2,3\}$ representing the
  door with the car behind it, and $V \in \{2,3\}$ the door opened by
  the quiz master; $\cZ = \{(1,2),(1,3),(2,3),(3,2)\}$ and for
  $z = (u,v)$, $U(z) = u$ and $V(z) = v$ (in this representation it is
  impossible for the quiz master to open a door with a car behind it).
  $\cP^*$ is the set of all distributions $P$ on $\cZ$ with uniform
  marginal $P(U)$, i.e. as usual, we assume the distribution of the
  car location to be uniform.  Let $\prag{P}$ be the conditional
  distribution for $U \mid V$ defined by
  $\prag{P}(U= 1 \mid V=2) = \prag{P}(U=1 \mid V=3) = 1/3$. This
  distribution can be arrived at using Bayes' theorem, starting with a
  particular $P \in \cP^*$, namely the $P$ with
  $P(V=2 \mid U=1) = P(V= 3 \mid U=1) = 1/2$, meaning that when the
  car is actually behind door $1$ and the quiz master has a free choice
  what door to open, he will flip a fair coin to decide. As this game
  was actually played on TV, it was in fact unknown whether the quiz
  master actually determined his strategy this way --- a quiz master
  who would want to be helpful to the contestant would certainly do it
  differently, for example choosing door $3$ whenever that is an
  option.  Nevertheless, most analyses, including Vos Savant's
  original one, assume this particular $\prag{P}$, and wars have been
  raging on the wikipedia talk pages as to whether this assumption is
  justified or not \citep{Gill11}.

  Interestingly, if we adopt this fair-coin $\prag{P}$ then $U' =
  \indicator_{U = 1}$ becomes a discrete simple pivot, in our sense,
  and $\prag{P}$ becomes pivotally safe, as is easily checked from
  Definition~\ref{def:pivot} and Definition~\ref{def:pivotal}: $f_v$
  in the definition is given by $f_2(1)=1, f_2(3) = 0, f_3(1)=1,
  f_3(2) = 0$ ($f_2(2)$ and $f_3(3)$ are undefined).  Thus $\prag{P}$
  is pivotally safe for Monty Hall and thus Theorem~\ref{thm:loss} can
  be applied, showing that, if DM takes decisions that are optimal
  according to $\prag{P}$, then these decisions will be exactly as
  good as she expects them to be for symmetric losses such as
  $0/1$-loss (as in the original Monty Hall problem) but also for the
  Brier and logarithmic loss. Relatedly, \cite{OmmenKFG16} shows that
  basing decisions on $\prag{P}$ will lead to admissible and minimax
  optimal decisions for {\em every\/} symmetric loss function (and, in
  a sequential gambling-with-reinvestment context, even when payoffs
  are asymmetric). This points to a general link between safety and
  minimax optimality, which we will explore in future work. Thus,
  while a strict subjective Bayesian would be {\em forced\/} to adopt
  a single distribution here --- for which we do not see very
  compelling grounds --- one can just adopt the uniform $\prag{P}$ for
  entirely pragmatic reasons: it will be minimax optimal and as good
  as one would expect it to be if it were true, even if it's in fact
  wrong --- it may, perhaps, be the case, that people have inarticulate intuitions in this direction and therefore insist that $\tilde{P}$ is `right'. 
\end{example}

\section{Beyond Conditioning; Beyond Random Variables}
\label{sec:general}
We can think of our pragmatic $\tilde{P}(U|V)$ as probability updating
rules, mapping observations $V=v$ to distributions on $U$. We required
these to be compatible with conditional distributions: $\tilde{P}(U
|V)$ must always be the conditional of {\em some\/} distribution
$\tilde{P}$ on $\cZ$, even though this distribution may not be in
$\cP^*$. Perhaps this is too restrictive, and we might want to
consider more general probability update rules. Below we indicate how
to do this --- and present Proposition~\ref{prop:above} which seems to
suggest that rules that are incompatible with conditional probability
are not likely to be very useful. We then continue to extend our
approach to update distributions given {\em events\/} rather than RVs,
leading to the `sanity check' we announced in the introduction.
For simplicity, in this section we restrict ourselves again to $V$
with countable range.
\begin{definition}\label{def:predictive}{\bf [Probability Update Rule]}
Let $U$ be an RV and $V$ be a generalized RV 
on $\cZ$ where $\range{V}$ is countable. 
A {\em probability update rule} $\pur{Q}(U \| V)$ 
is a function from $\range{V}$ to the set
  of distributions on $\range{U}$. We call $\pur{Q}(U \|V)$ {\em logically coherent\/} if, for each $v \in
  \range{V}$, the corresponding distribution on $\range{U}$,
  denoted $\pur{Q}(U \| V=v)$, satisfies
\begin{equation}\label{eq:compatibility}
\pur{Q}(U \in \{ u: (u,v) \in \range{(U,V)} \} \| V = v) = 1.\end{equation}
We call $\pur{Q}(U \| V)$ {\em compatible with conditional probability\/}
if there exists a distribution $P$ on $\cZ$ with full support for $V$ ($\support_{P}(V)= \range{V}$) such that $\pur{Q}(U \| V) \equiv P(U \mid V)$. 
\end{definition}
Logical coherence is a weak requirement: if RVs $U$ and $V$ are {\em
  logically separated}, i.e.  $\range{(U,V)} = \range{U} \times
\range{V}$ (as is the case in all examples in this paper except
Example~\ref{ex:montyhallb}) then clearly every, arbitrary function
from $\range{V}$ to the set of distributions on $\range{U}$ is a
logically coherent probability update rule.  However, if
$\range{(U,V)} \neq \range{U} \times \range{V}$, then there are
logical constraints between $U$ and $V$. For example, we may have $\cZ
= \{1,2\}$, and $U(z) = V(z) = z$ (so that $U$ and $V$ are
identical). Then a probability update rule $\pur{Q}$ with $\pur{Q}(U=1
\| V=2) = 1$ would be logically incoherent.  Every rule that is
compatible with conditional probability is logically coherent; there
does not seem much use in using logically incoherent rules.

For given $\pur{Q}(U \| V)$ we can now define safety for $U|V$, $U
|\ave{V}$ and calibration for $U|V$ as before, using
Definition~\ref{def:leftsafety}, \ref{def:maina} and~\ref{def:calibration}. Note however that
notions like safety for $U \mid \cali{V}$ and 
pivotal safety are not defined, since these are defined
in terms of marginal distributions  for $U$ (or $U'$,  respectively), and the marginal $\pur{Q}(U)$ is undefined for probability update rules $Q$.
\begin{proposition}\label{prop:above}
  If $\tilde{Q}(U \| V)$ is not compatible with conditional probability, then
  it is not safe for $U \mid \ave{V}$ (and hence, as follows directly from the hierarchy of Figure~\ref{fig:figure},  also unsafe for $U
  \mid V$, and also not calibrated for $U \mid V$).
\end{proposition}
\begin{proof} Follows directly from the characterization of safety for
  $U \mid \ave{V} $ given in
  Proposition~\ref{prop:newstart}. \end{proof} This result suggests
that rules that are incompatible with conditional probability are not
likely to be very useful for inference about $U$; the result says
nothing about the weaker notions of safety with $\ave{U}$ rather than
$U$ on the left though, or with $\wksafe{V}$ instead of $\ave{V}$ or $V$ on the left.  

\paragraph{Conditioning based on Events}
Suppose we are given a finite or countable set of outcomes $\cU$ with a
distribution $P_0$ on it, as well as a set $\cV$ of nonempty subsets
of $U$. We are given the information that the outcome is contained in the set 
$v$ for some $v \in
\cV$, and we want to update our distribution $P_0$ to some new
distribution $P'_0(\cdot \| v)$, taking the information in $v$ into
account. A lot of people would resort to {\em naive conditioning\/} here
\citep{GrunwaldH03}, i.e. follow the definition of conditional
probability and set $P'_0$ to $P_{\text{naive}}(\{u \} \mid v) :=
P_0(\{u\})/P_0(v)$. We want to see whether such a $P_{\text{naive}}$
is {\em safe}. To this end, we must translate the setting to our
set-up: to make a probability update rule in our sense well-defined
(Definition~\ref{def:predictive}), we must have a space $\cZ$ on which
the RV $U$, denoting the outcome in $\cU$, and $V$, denoting the
observed set $v$, are both defined. To this end we call {\em any\/}
set $\cZ$ such that, for all $u \in \cU, v \in \cV$ with $u \in v$,
there exists a $z \in \cZ$ with $U(z) = u$ and $V(z) = v$, a set {\em
  underlying $\cU$ and $\cV$\/} (we could take $\cZ = \cU \times \cV$,
but other choices are possible as well). We then set $\cP^*$ to be the
set of all distributions $P$ on $\cZ$ with marginal distribution $P_0$
on $U$ and, for all $v \in \cV$, $P(U \in v \mid V=v) = 1$. We may now
ask whether the naive update,
\begin{equation}\label{eq:naive}
\prag{Q}(U= u \| V=v) :=P_{\text{naive}}(\{u \}\mid v)
\end{equation} is safe. The
following proposition shows that in general it is not:
\begin{proposition}{\ \bf [\cite{GrunwaldH03}, rephrased]}
For given $P_0$, $\cU$ and $\cV$, let $\cZ$ be any set underlying $\cU$ and $\cV$ and let $\cP^*$ be the associated set of distributions compatible with $P_0$.
We have: $\prag{Q}(U \| V)$ defined as in (\ref{eq:naive}) is the conditional of some distribution $\prag{Q}$ on $\cZ$ that is safe for $U \mid V$ if and only if $\cV$ is a 
partition of $\cU$.\end{proposition}

If $\cV$ is not a partition of $\cU$, then in some cases
$\prag{Q}(U\|V)$ is still compatible with conditional probability;
then it is still potentially safe for $U |V$; in other cases it is not even
compatible with conditional probability and hence by
Proposition~\ref{prop:above} guaranteed to be unsafe. The main result
of \cite{GillG08} can be re-interpreted as giving a precise
characterization of when this guaranteed unsafety holds.

To illustrate, consider Monty Hall, Example~\ref{ex:montyhall}
again. In terms of events, $\cU = \{1,2,3\}$ and $\cV = \{ \{1,2 \},
\{1,3\}\}$: if Monty opens door $x, x \in \{2,3\}$, then the event
`car behind door $1$ or $x$', i.e. $\{1,x\}$ is observed, so
$P_{\text{naive}}(\{1\} \mid \{1,2\}) = P_{\text{naive}}(\{1\} \mid
\{1,3\}) = 1/2$, leading to the common false conclusion that the car
is equally likely to be behind each of the remaining closed
doors. Clearly though, $\cV$ is {\em not\/} a partition of $\cU$,
since it has overlap, so by the proposition, $P_{\text{naive}}$ is {\em
  unsafe\/} for $U \mid V$. Intuitively, it is easy to see why: if
$U=1$, the quiz master has a choice what element of $\cV$ to present,
and may do this by flipping a coin with bias $\theta$. Therefore the
set $\cP^*$ has an element $P_{\theta}$ corresponding to each $\theta
\in [0,1]$, and the correct conditional distribution $P_{\theta}(U=1 \mid V=v)$ depends on
$\theta$ in a crucial way (and will in fact not be equal to
$\prag{Q}$, no matter the value of $\theta$). But DM need not be
concerned with any of these details: what matters is that naive
conditioning is not safe, which, by Proposition~\ref{prop:above}, is immediate from the fact that $\cV$ is not a partition of $\cU$. 

  The fact that conditioning is problematic
if one conditions on something not equal to a partition has in fact
been known for a long time, see e.g. \cite{Shafer85} for the first
landmark reference. Our point is simply to show that the issue
fits in well with the safety concept.  There is an obvious analogy
here with the Borel-Kolmogorov paradox \citep{schweder1996bayesian}
which presumably could also be recast in terms of safety. As
\cite{Kolmogorov33} writes, `` The concept of a conditional
probability with regard to an isolated hypothesis whose probability
equals 0 is inadmissible.'' Safe probability suggests something more
radical: standard conditional probabilities with regard to an isolated
hypothesis (event) are {\em never\/} admissible --- if one does not
know whether the alternatives form a partition, setting $\prag{P}$ to
be the standard conditional distribution is inherently unsafe.

\section{Parallel, Earlier and Future Work; Open Problems and Conlusion} 
\label{sec:conclusion}
\subsection{Parallel Work: Safe Testing}
There is one application of safe probability that has so many
implications that we decided to devote a separate paper to it, which
we hope to finish soon. This is the use of safety concepts in {\em
  testing}, already alluded to in the introduction. Here, let us just
very briefly outline some main ideas. Consider a testing problem where
we observe data $X^N$ and $h_0$ stands for a null hypothesis which
says that data are a sample from some $P_0$ belong to a statistical
model $\cM_0$. For simplicity we will only consider the case of a
point null hypotheses in this mini-overview, so $\cM_0 = \{P_0\}$ is a
singleton. $h_1$ represents another set of distributions $\cM_1$,
which may, however, be exceedingly large --- in fact it may impossible
for us to state it exactly, for it may be, for example, as broad as
`the data are a sample of text in some human language unknown to
us'. We associate $h_0$ with distribution $P_0$ and corresponding
density or mass function $p_0$, and $h_1$ with some {\em single\/}
distribution $P_1$ with associated $p_1$. If $\cM_1$ is a parametric
model, or a large but still precisely specifiable model such as a
nonparametric model, then we might take $P_1$ to be the Bayes marginal
distribution under some prior $\Pi$, $p_1(X^n) := \int p(X^n)
d\Pi(p)$, but other choices are possible as well, and may sometimes
even be preferable.

We now define a `posterior' $\tilde{P}(H \mid X^N)$ by setting 
\begin{equation}\label{eq:pseudopost}
\tilde{P}(H= h_0 \mid X^N) := \frac{p_0(X^n)}{p_0(X^n)+ p_1(X^n)},
\end{equation}
which would coincide with a standard Bayesian posterior based on prior
$(1/2,1/2)$ if we used a $p_1$ set in the Bayesian way described
above. In that special case it also broadly corresponds to the method
introduced by \cite{BergerBW94} that, in its culminated form
\citep{Berger03} provides a testing method that has a valid
interpretation within the three major testing-schools: Bayes-Jeffreys,
Neyman-Pearson and Fisher. Readers familiar with the MDL (minimum
description length) paradigm \citep{Grunwald07} will notice that for
every complete lossless code for $X_1,X_2, \ldots$ that encodes $X^N$
with $L(X^n)$ bits, setting $p_1(X^n) = 2^{-L(X^x)}$ provides a
probability mass function on sequences of length $n$. Thus, if one has
a code available which one thinks might compress data well, one can
set $p_1$ in this non-Bayesian way. The log-posterior odds $\log
\tilde{P}(h_0 \mid X^N)/\tilde{P}(h_0 \mid X^N) = - \log p_1(X^n) +
\log p_0(X^n)$ then have an interpretation as the {\em number of bits
  saved by compressing the data with the code $L$ compared to the code
  that would be optimal under $P_0$\/}; thus, the approach of
\cite{ryabko2005using} neatly fits into this framework; so does the
Martingale testing approach of \cite{Vovk93} in which $p_1$ is
determined by a sequential gambling strategy; for any gambling
strategy $g$, there is a corresponding probability mass function $p_1$
such that the inverse of the (pseudo-) `Bayes factor'
\begin{equation}\label{eq:bf}
  \frac{\tilde{P}(h_0 \mid X^N)}{\tilde{P}(h_0 \mid X^N)} = 
  \frac{p_0(X^n)}{p_1(X^n)}\end{equation} 
can be
interpreted as the amount of money gained by gambling strategy $g$
under pay-offs that would be fair (yield no gain in expectation) if
the null $P_0$ were true: $p_1(X^n)/p_0(X^n)$ is the factor by which one's initial capital is multiplied if one gambles according to $g$ under odds that are fair under $h_0$, so that the more money gained, the larger the evidence against $h_0$. For an example of useful non-Bayesian gambling strategies (or equivalently, distributions $p_1$) we refer to the {\em switch distribution\/} of \cite{ErvenGR07}. 

\paragraph{Where Safety Comes In}
One can now base inferences on $\tilde{P}(H \mid X^N)$ just as a
Bayesian would --- with the essential stipulation that one only does
this for the subset of inferences that once considers {\em safe\/} in
the appropriate sense. For example, suppose that the data compressor
{\tt gzip} compresses our sequence of data substantially more than our
null hypothesis $P_0$, which says that the outcomes are
i.i.d. Bernoulli$(1/2)$.  Thus, $\tilde{P}(H = h_0 \mid X^N)$ will be
exceedingly small, yet the $p_1$ corresponding to {\tt gzip} may
certainly not be considered `true'. We thus do not want to take the
predictions made by $p_1$ for future data $X_{N+1},\ldots$ too
seriously. We can accomplish this by declaring that $p_1$ is {\em
  not\/} safe for $X_{N+1} \mid X^N$ relative to $\cP^* \mid H=
h_1$. Note that we can declare such unsafety without actually
precisely specifying $\cP^*| H = h_1$, which may be too complicated to
do.  On the other hand, if we do believe that $p_1$ accurately
describes our knowledge of $\cM_1$, e.g. $\cM_1$ is small and $p_1$ is
a Bayes marginal distribution based on substantial prior knowledge
codified into $\Pi$, then we can declare $p_1$ to be safe relative to
$\cP^* \mid H= h_1$.  We thus have a single framework that encompasses
both the Fisherian (falsificationist) and the Bayes/Neyman-Pearson
testing paradigms, depending on what inferences we consider safe. On a
technical level however, this framework avoids many difficulties of
the standard implementations of the Fisherian and the Neyman-Pearson
paradigms. Compared to Fisher, we avoid the use of $p$-values
(although the `Bayes factor' (\ref{eq:bf}) can be interpreted as a
{\em robustified, sample-plan independent\/} $p$-value
\citep{shafer2011test,PasG14}).  We consider this a very good thing in
the light of the many difficulties surrounding $p$-values such as (to
mention just two) their dependence on the sampling plan, making them
impossible to use in many simple situations and their interpretation
difficulties \citep{Berger03,Wagenmakers07}. Compared to
Neyman-Pearson's original formalism, we do not just get an `accept' or
`reject' decision, but also a measure of evidence (the (pseudo-) `Bayes factor')
that can be used to infer stronger conclusions as we get stronger
evidence --- in contrast to conclusions based on $p$-values, such
conclusions often remain {\em safe\/} in the appropriate sense.

Indeed, in the second part of this work we consider safety of
$\tilde{P}(H \mid X^N)$ in terms of loss functions $L(H,\delta(X^N))$
where $H \in \{h_0,h_1\}$ and $\delta(X^N)$ is the Bayes act based on
$\tilde{P}(H \mid X^n)$. In the simplest case takes $\delta$ takes
values in the decision set $\cA= \{ {\tt accept},{\tt reject} \}$. We
find that, under some conditions on $p_1$, $\tilde{P}(H \mid X^n)$ is
safe for $L(H,\delta(X^N)) \mid \wksafe{X^N}$, i.e. we have safety in
the weakest (but still useful) sense defined in this paper. While
standard Type-I and Type-II error guarantees of the Neyman-Pearson
approach can be recast in this way, safety continues to hold if $\cA$
has more than two elements with different losses associated --- a
realistic situation which cannot be handled by either a Neyman-Pearson
or a Fisherian approach. In this situation, making the right decision
means one has to take the strength of evidence into account --- if
there is more evidence against $h_0$, then the best action to take
will have lower loss under $h_1$ but higher loss under $h_0$. As soon
as there are more than two actions, measuring evidence in terms of
unmodified $p$-values leads to unsafe inferences; yet inferences based
on the (pseudo)-posterior tend to remain safe.

Furthermore, we can also check whether we retain safety under {\em
  optional stopping\/} \citep{BergerBW94,shafer2011test,PasG14}. We
find that, under further conditions on $p_0$ and $p_1$, $\tilde{P}$
remains safe for $L(H,\delta(X^N)) \mid \wksafe{X^N}$, even though $N$
is now a RV (determined by the potentially unknown stopping rule) with
an unknwon distribution. Interestingly, things get even better with a
slight modification of $\tilde{P}$, where we set $\tilde{P}(H = h_0 \mid
X^N)$ to $\max\{1, p_0(X^N)/p_1(X^N)\}$, i.e. we use the posterior {\em odds\/}
{\em as\/} the posterior {\em probability}. With this `posterior' we automatically get (weak)
safety under arbitrary optional stopping and for essentially arbitrary
loss functions -- no more conditions on $p_0$ and $p_1$ are
needed. The reason is that with this choice $\tilde{P}(H \mid X^n)$
becomes bounded by a test martingale in the sense of
\cite{Vovk93} and \cite{shafer2011test}. If we want to use the standard posterior
as \cite{Berger03} does, we either need to change the action $\delta$
a little (introducing a so-called `no-decision region', as also done by \cite{BergerBW94}) or make strong
assumptions about $p_0$ and $p_1$.

It is often claimed that optional stopping is not a problem for
Bayesians, since the Bayesian inferences do not depend on the sampling
plan. For objective Bayesian inference, this is incorrect (priors such
as Jeffreys' do depend on the sampling plan); for subjective Bayesian
inference, this statement is correct only if one really fully believes
one's own subjective prior. As soon as one uses a prior partially for
convenience reasons --- which happens in nearly all practical
scenarios --- validity of the conclusions under optional stopping is
compromised. Safe testing allows one to establish validity under
optional stopping, in a `weak safety' sense, even in such cases ---
essentially, one's conclusions will be safe under optional stopping
under any $P$ in the set $\cP^*$ of possible distributions, not just
the single distribution one adopts as a subjective Bayesian.

\subsection{Earlier Work and  Future Work}
The idea that fiducial or confidence distributions can be used for
valid assessment of some, not all, RVs or events that can be defined
on a domain has been stressed by several authors,
e.g. \cite{SchwederH02,xie2013confidence,Hampel06}. The novelty here
is that we formalize the idea and place it in broader context and
hierarchy.  The idea of replacing sets of distributions by a single
representative also underlies the MDL Principle \citep{DeRooijG09},
yet again, without broader context or hierarchy. It is also the core
of the pignistic transformation advocated by \cite{Smets89}
as part of his {\em transferable belief model}, which, apart from the
transformation idea, seems to be almost totally different from safe
probability however --- it would be interesting to sort out the
relations though. I already noted in the introduction that my own
earlier work contains various definitions of partial notions of
safety, but unifying concepts, let alone a hierarchy, were lacking.

\paragraph{Future Work I: Safety vs. Optimality}
There is one crucial issue though that we neglected in this paper, and
that was brought up earlier, to some extent, by \citep{Grunwald00a}
and \citep{GrunwaldH04}: the fact that mere safety isn't enough --- we
want our pragmatic $\prag{P}$ to also have optimality properties (see
e.g. Example~\ref{ex:calibration} for the trivial weather forecaster
who is calibrated (safe) without being optimal). As indicated by
Example~\ref{ex:montyhallb} in the present paper, and also by
\citep{Grunwald00a} and more implicitly by \cite{OmmenKFG16}, there is
a link between safety and minimax optimality which sometimes can be
exploited, but much remains to be done here --- this is our main goal
for future work.

\paragraph{Future Work II: Objective Bayes} Safe Probability may also
be fruitfully be applied to objective Bayesian inference. For example,
consider inference of a Bernoulli parameter based on an `objective'
Jeffreys' prior
$\pi(\theta) \propto \theta^{-1/2} (1- \theta)^{-1/2}$. Use of such a
prior may certainly be defensible because of its geometric and
information-theoretic properties \citep{DeRooijG09}, but what if we
have a very small sample of just 1 or even 0 outcomes?  Then Jeffreys'
prior would tell us, for example, that a bias $\theta$ between $0$ and
$0.01$ is 10 times as likely than a bias between $0.495$ and
$0.505$. Most objective Bayesians would probably not be prepared to
gamble on that proposition.\footnote{One might object that an actual
  value of $\theta$ may not even exist, and certainly will never be observed,
  so one cannot gamble on it. But I could propose this gamble instead:
  I will toss the biased coin 10000 times, and only reveal to you the
  final relative frequency of heads. How much would you bet on it
  being $\leq 0.01$?}  This is fine, but then what propositions would
an objective Bayesian be prepared to gamble on, and what not? Bayesian
inference has no tools to deal with this question --- and--- in a
manner similar to characterization of safety for fiducial
distributions --- safe probability may offer them.
\paragraph{Future Work III: Epistemic Probability} More generally,
both objective Bayesian and fiducial methods have been proposed as
candidates for {\em epistemic probability\/}
\citep{Keynes21,Carnap50,Hampel06} but it is unclear how exactly
such a notion of probability should be connected to decision theory
--- while a Bayesian or frequentist probability of $0.01$ on outcome
$A$ implies that a (not too risk-averse) DM would be willing to pay
one dollar for a lottery ticket that pays off 200 dollar if $A$ turns
out to be the case, for epistemic probability this is not so
clear. Safe probability suggests that it might be fruitful to view
epistemic probabilities as assuming a willingness to bet on a {\em
  strict subset\/} of all events $\cA$ that can be defined on the given
space.

\paragraph{Open Problems}
Other future work involves open problems, as mentioned in the caption
of Figure~\ref{fig:figure}. Of particular interest is whether we can
extend confidence safety to multidimensional $U$.  Earlier work
\citep{DawidS82,Seidenfeld92} suggests that then in general, there
will be multiple, different choices for $\tilde{P}$, none of which is
inherently `best'. A major additional goal for future work is to
identify subjective considerations that may lead one to prefer one
choice over another, cf. the idea of `luckiness'
\citep{DeRooijG09}. Another intriguing question is whether safety can
be re- construed as an {\em extension\/} of measure theory --- which
has also been designed to restrict the notion of (probability)
measures so that they cannot just be applied to any set one likes.
Yet another avenue is to extend the definition of pragmatic
distributions using upper- and lower expectations, replacing
$\tilde{P}$ by a set of distributions $\tilde{\cP}$ (this is briefly
detailed in Appendix~\ref{app:partial}).  Then both $\prag{\cP}$ and
$\cP^*$ would fall into the `imprecise probability' paradigm; we could
still get nontrivial predictions as long as $\prag{\cP}$ is more
`specific' than $\cP^*$. Such an extension would hopefully allow us to
represent the random-set approach to fiducial inference from
\cite{Dempster68} and its modern extensions, such as the inferential
models of \cite{martin2013inferential}, as an extension of pivotal
safety. Here confidence-safe probabilities would be replaced by
confidence-safe probability intervals; perhaps one could even arrive
at a general description of what applications of Dempster-Shafer
theory \citep{Shafer76,Dempster68} are safe at all, and if so, to what
degree they are safe.

\paragraph{Acknowledgements}
This manuscript has benefited a lot from various discussions over the
last fifteen years with Philip Dawid, Joe Halpern and Teddy
Seidenfeld. Special thanks go to Teddy as well as to Gert de Cooman
and Nils Hjort for providing encouragement that was essential to get
this work done.  This research was supported by the Netherlands
Organization for Scientific Research (NWO) VICI Project Nr.
639.073.04.
\bibliography{master,peter,MDL}
\appendix
\section{Technical Extras and Proofs}\label{app:technical}
\subsection{Details for Section~\ref{sec:pstar}: partially specified $\prag{P}$}\label{app:partial}

As promised in the main text, here we consider $\prag{P}$ that are
only partially specified. We may think of
these again as {\em sets\/} of distributions, just as we do for
$\cP^*$.  For example, consider an update rule $\pur{Q}(U \| V)$ as in
Definition~\ref{def:predictive} that is compatible with conditional
probability. Such a $\pur{Q}$ is a prime example of a {\em partially
  specified pragmatic distribution\/}: it is the conditional of at
least one distribution $P$ on $\cZ$, but there may (and will) be many
more, different $P$ for which it is also the conditional. We may thus
associate $\pur{Q}$ with the (nonempty) set $\pragset{\cQ}$ of all such distributions
$P$ on $\cZ$ with $P(U \mid V)= \pur{Q}(U \| V)$.  Then clearly, for every
RV $U'$ on $\cZ$ with $(U,V) \determines U'$, all $Q_1, Q_2 \in \pragset{\cQ}$,
we have $Q_1(U' |V) = Q_1 (U' |V)$; thus the distribution of such
$U_1|V$ is determined by $\pragset{\cQ}$; but for $U'$ not determined by
$(U,V)$, there may be $Q_1, Q_2 \in \pragset{\cQ}$ with
$Q_1(U' |V) \neq Q_2(U'|V)$ and we may have to make  assessments about $U'$ given $V$ in terms of lower- and upper-expectation intervals  
$[\inf_{Q \in \pragset{\cQ}} \Exp_{Q}[U \mid V], \sup_{Q \in \pragset{\cQ}} \Exp_{Q}[U \mid V]$. 

A more involved calculation shows that for all $Q_1, Q_2 \in
\pragset{\cQ}$, we have $Q_1(U |V') = Q_2(U|V')$ iff $V \determines V'
\determines Q(U\|V)$; a condition that also plays a role in
Theorem~\ref{thm:first} on calibration. One might thus try to state
and prove restricted versions of our results, holding for partially
specified $\pur{Q}$ of this form. In practice though, one also
encounters other types of partially specified ${Q}$ (for example, in
regression contexts the function $\Exp_{Q}[U |V]$ might be used, but
no other aspect of $Q$ is relevant). It might thus be more useful to
generalize the whole machinery to arbitrary sets of distributions
$\pragset{\cQ}$; an additional potential advantage is that this might
allow us to determine safety of inferential procedures that output
sets of probabilities that are nonsingleton yet avoid dilation, such
as the inferential model approach of \cite{martin2013inferential} and
more generally Dempster-Shafer theory.  To get a first idea of how
this might work, consider the second part of the basic
Definition~\ref{def:leftsafety}. Here we essentially only have to
change $\prag{P}$ to $\pragset{\cP}$; nothing else changes:
\begin{definition}\label{def:leftsafetyb}
  Let $\cZ$ be an outcome space and $\cP^*$ and $\pragset{\cP}$ be sets of distributions
  on $\cZ$ as defined in Section~\ref{sec:pstar}, let $U$ be an RV
  and $V$ be a generalized RV on $\cZ$.  We say that $(\cP^*, \pragset{\cP})$ is {\em
sharply    safe\/} for $\ave{U}_{\we} \mid \ave{V}_{\we}$ 
  if 
\begin{equation}\label{eq:startsafeagain}
\text{for all $P^* \in \cP^*, \prag{P} \in \pragset{\cP}$}: \ \Exp_{P^*}[U] = \Exp_{P^*}[ \Exp_{\prag{P}} [U | V]].
\end{equation}
\end{definition}
All other definitions may be changed accordingly.  We call the
resulting notions `sharply' safe because it requires, for example,
safety for ${U} \mid {V}$ to imply that {\em all\/} distributions in
$\pragset{\cP}$ agree on $\prag{P}(U \mid V)$, i.e. their conditional
distributions of $U \mid V$ are the same; one could also define weaker notions in which this is only required for {\em some\/} $\tilde{P} \in \tilde{\cP}$. 

\subsection{Details for Section~\ref{sec:four}}
\paragraph{Proof of Proposition~\ref{prop:newstart},}
Let $k$ be such that $\range{U} \subset \reals^k$.

{\em Part 1}. 
Safety of $\prag{P}$ for $U \mid \ave{V}$ implies that for every
vector $\vec{a} \in \reals^k$ the RV $U_{\vec{a}} = \indicator_{U \leq \vec{a}}$ 
satisfies, for all $P \in \cP$, 
\begin{equation*}
\Exp_{V \sim P} \Exp_{U_{\vec{a}} \sim P \mid V} [U_{\vec{a}}] = 
\Exp_{V \sim P} \Exp_{U_{\vec{a}} \sim \prag{P} \mid V} [U_{\vec{a}}],
\end{equation*}
which can be rewritten as 
$$
\sum_{v \in \range{V}} P(V=v) [P(U \leq {\vec{a}} \mid V= v)] =  
\sum_{v \in \range{V}} P(V=v) [\prag{P}(U \leq {\vec{a}} \mid V= v)],
$$
which in turn is equivalent to $P(U \leq \vec{a}) = P'(U \leq \vec{a})$ with $P'$ as in the statement of the proposition.  This shows that safety for $U \mid \ave{V}$ implies (\ref{eq:startsafec}). Conversely, (\ref{eq:startsafec}) implies that for any function RV $U'$ with $\range{U'} \subset \reals^{k'}$ with $U \determines U'$, letting 
$f$ be the function with $U' = f(U)$, we have for every $P \in \cP^*$,
$$
\Exp_{U \sim P}[f(U)] =
\Exp_{U \sim P'}[f(U)] = 
\Exp_{V \sim P} \Exp_{U \sim \prag{P} \mid V}[f(U)],
$$
which implies that $\prag{P}$ is safe for $\ave{U'} \mid \ave{V}$,
Since this holds for every $U'$ with $U \determines U'$, safety for $U
\mid \ave{V}$ follows.

{\em Part 2\/} is just definition chasing. {\em Part 3\/} follows as a special case of Part 4 with $W$ in the role of $V$ and $V  \equiv {\bf 0}$. The if-part of  {\em Part 4\/} is a straightforward consequence of the definition. For the only-if part, note that 
from the definition of safety for $U \mid [V],W$ we infer that it implies that 
for all $P \in \cP$, for all $v \in \support_{\tilde{P}}(V)$, $w \in \support_P(W)$,
 for all functions $f$ and RVs $U' = f(U)$,
\begin{equation}\label{eq:yahoob}
\Exp_{U \sim P \mid
  W=w}[f(U)] =  \Exp_{U \sim \prag{P} \mid V=v,W=w} [f(U)].
\end{equation}
In
particular,  this will hold for every $\vec{a} \in \reals^k$, for the RV
$U_{\vec{a}} = f_{\vec{a}}(U) = \indicator_{U \leq \vec{a}}$. Then
(\ref{eq:yahoob}) can be written as $P(U \leq \vec{a} \mid W=w) =
\prag{P}(U \leq \vec{a} \mid V=v,W=w)$. Thus, the cumulative distribution
functions of $P(U \mid W=w)$ and $\prag{P}(U \mid V=v,W=w)$ are equal at
all $\vec{a} \in \reals^k$, so the distributions themselves must also
coincide, and (\ref{eq:calibrationagain}) follows.

\subsection{Details for Section~\ref{sec:calibration}}
\paragraph{Proof of Proposition~\ref{prop:ignore}}
We let $f_0$ be the function such that $V
\fPdetermines{f_0}{\prag{P}} P(U \mid V)$ and we let $f_1 = f$ be such that $V
\fPdetermines{f_1}{\prag{P}} V'$ ($f_0$ exists by definition, $f_1$ by
assumption). We also let $V'' \equiv \prag{P}(U \mid V)$ and note that every
$v'' \in \range{V''}$ is a probability distribution on $U$.

We first establish {\em (1) $\Leftrightarrow$ (2)}. 
For this, note that since $V \fPdetermines{f_1}{\prag{P}} V'$, we have for all $v \in \range{V}$, for the $v'\in \range{V'}$ with $f_1(v)= v'$, that
\begin{equation}\label{eq:nescafe} 
\prag{P}(U \mid V = v,V' = v') = \prag{P}(U \mid V=v).
\end{equation}
If (1) holds, i.e.  $\prag{P}(U \mid V,V')$ ignores $V$, then the
left-hand side in (\ref{eq:nescafe}) is equal to $\prag{P}(U \mid V'=
v')$, and (2) follows by plugging this into
(\ref{eq:nescafe}). Conversely, if (2) holds, then the right of
(\ref{eq:nescafe}) is equal to $\prag{P}(U \mid V'= v')$ for all $v$ with $f_1(v) = v'$, and (2)
follows by plugging this into (\ref{eq:nescafe}).

{\em (2) $\Rightarrow$ (3)\/} Suppose that (2) holds. This immediately
implies that $V' \fPdetermines{f_2}{\prag{P}} \prag{P}(U \mid V)$ with $f_2(v')
= \prag{P}(U \mid V'= v')$, which is what we had to prove.

{\em (3) $\Rightarrow$ (4)} Suppose that (3) holds. We may thus assume
that $V'\fPdetermines{f_2}{\prag{P}} V'' (\; \equiv \prag{P}(U \mid V)\;)$ for some function
$f_2$. By equivalence $(1)\Leftrightarrow (2)$, which we already
proved, it is sufficient to show that for all $v'' \in \range{V''}$,
for all $v' \in \range{V'}$ with $f_2(v')= v''$, we have $\prag{P}(U
\mid V=v')= \prag{P}(U \mid V'' = v'')$. Since $
\prag{P}(U \mid V'' = v'') = v''$, 
it is sufficient to prove that for all $v'' \in \range{V''}$ and
for all $v' \in \range{V'}$ with $f_2(v')= v''$, we have $\prag{P}(U
\mid V=v')= v''$.

To prove this, fix arbitrary $v'' \in \range{V''}$. For all $v'\in
\range{V'}$ with $f_2(v') = v''$, for all $v \in \range{V}$ with
$f_1(v)= v'$, we must have $f_2(f_1(v))= v''$ and hence $f_0(v) =
v''$, so (by definition of $V'' $) $\prag{P}(U \mid V=v) = f_0(v) =
v''$. Since (from the fact that $V \Pdetermines{\prag{P}} V'$ and the definition
of conditional probability) we can write $\prag{P}(U \mid V= v') =
\sum_{v \in \range{V}: f_1(v) = v'} \prag{P}(U \mid V = v) \alpha_v$
for some weights $\alpha_v \geq 0$, $\sum_{v: f_1(v)= v'} \alpha_v =
1$, and all components of the mixture must be equal to $v''$, it
follows that $\prag{P}(U \mid V = v') = v''$, which is what we had
to prove.

{\em (4) $\Rightarrow$ (2)} We may assume that $V'\fPdetermines{f_2}{\prag{P}}
V'' \equiv \prag{P}(U \mid V)$ for some function $f_2$, and (by
equivalence $(1) \Leftrightarrow (2)$ which we already established)
that for $v'' \in \range{V''}$, all $v' \in \range{V'}$ with $f_2(v')
= v'' $, $\prag{P}(U \mid V' = v') = \prag{P}(U \mid V'' = v'')$. By
definition of $V''$, the latter distribution must itself be equal to
$v''$, so we get:
\begin{equation}\label{eq:robbie} \prag{P}(U \mid V' = v') = v'',
\end{equation}
We must also have, for all $v$ with $f_1(v) = v'$, that $f_2(f_1(v)) =
v''$, so $f_0(v) = v''$, so, by definition of $V''$, $\prag{P}(U \mid
V=v) = v''$. Combining this with (\ref{eq:robbie}) gives that
$\prag{P}(U \mid V=v) = \prag{P}(U \mid V'= v')$, and, because $V
\Pdetermines{\prag{P}} V'$, that $\prag{P}(U \mid V=v, V'= v') = \prag{P}(U
\mid V'= v')$.  This must hold for all $v'' \in \range{V''}$, all $v'
\in \range{V'}$ with $f_2(v') = v''$, and hence simply for all $v'\in
\range{V'}$ and hence $\prag{P}(U \mid V, V')$ ignores $V$.

{\em Final Part\/} 
By Equivalence $(1) \Leftrightarrow (2)$, we have for all $v' \in \range{V'}$, all $v \in \range{V }$ with $f_1(v) = v'$, that $\prag{P}(U \mid V= v) = \prag{P}(U \mid V' = v')$. Combining this equality with the assumed 
safety of $\prag{P}$ for $U \mid V$, we must also have, for all $P \in \cP^*$,
all $v \in \range{V}$ with $f_1(v) = v'$, that 
\begin{equation}\label{eq:poekie}
{P}(U \mid V= v) = \prag{P}(U \mid V' = v'),
\end{equation}
But since $P(U \mid V'= v')$ must be a mixture of $P(U \mid V=v)$
over all $v$ with $f_1(v) = v'$ (as in the proof of $(3)
\Rightarrow (4)$ above), and all these mixture components are
identical by (\ref{eq:poekie}), we get that $P(U \mid V'= v') =
\prag{P}(U \mid V'= v')$. Since this argument is valid for all $v'\in
\range{V'}$, we have established safety for $U \mid V'$.

\paragraph{Proof of Theorem~\ref{thm:first}}
The result $(1) 
\Leftrightarrow (3)$ is almost immediate: calibration of $\prag{P}$ is equivalent to having, for each $P \in \cP^*$, for each $v'' \in \support_P(V'')$, (note that each such $v''$ is a probability distribution on $U$): 
$$P(U \mid \prag{P}(U \mid V_0) = v'') = v'' =
\prag{P}(U \mid V''= v'').$$ Rewriting the expression on the right of
the leftmost conditioning bar as $V'' = v''$, we see that this is
equivalent to having
$$
P(U \mid V'' = v'') =_P
\prag{P}(U \mid V''= v'')
$$
which by Proposition~\ref{prop:newstart} is equivalent to 
safety for $U \mid V''$ 
and so $(1) \Leftrightarrow (3)$ follows. From the definition of safety for $U \mid [V],V'$, Definition~\ref{def:leftsafety}, $(3) \Rightarrow (2)$ now
follows if we can show (by taking, in (2), $V'= V''= \prag{P}(U \mid V)$), 
that (a) $V \Pdetermines{\prag{P}} V''$ and (b) $\prag{P}(U \mid V,V'')$ ignores $V$. The first requirement holds trivially, the second  follows from Proposition~\ref{prop:ignore}, $(3) \Rightarrow (1)$, taking again $V'\equiv \prag{P}(U \mid V)$ (so that automatically $V \determines V'$ and $V' \determines \prag{P}(U \mid V)$). 

It now only remains to show $(2) \Rightarrow (3)$. So suppose that
$\prag{P}$ is safe for $U \mid V'$ and $\prag{P}(U \mid V,V')$ ignores $V$ and $ V \Pdetermines{\prag{P}} V'$. By Proposition~\ref{prop:ignore} $(1) \Rightarrow (4)$, it
follows that $\prag{P}(U \mid V',V'')$ ignores $V'$, where $V'' \equiv
\prag{P}(U \mid V)$. The result now follows by the final part of
Proposition~\ref{prop:ignore}, applied with $V$ in the proposition set
equal to $V'$ and $V'$ in the proposition set equal to $V''$.

\subsection{Details for Section~\ref{sec:continuous}}
\label{app:continuous}
\paragraph{Proof of Theorem~\ref{thm:confidence}}

{\em (1) $\Leftrightarrow$ (2)}.  First assume $U'$ is a simple pivot and that pivotal safety holds for $U'$. Set $f_v(u)$ as in
Definition~\ref{def:pivot} and take it to be increasing for each $v
\in \range{V}$ (the decreasing case is analogous). Since $U'$ is a
pivot and pivotal safety holds, we have, for all $v \in \range{V}$, $u' \in \range{U'}$,
$\prag{F}_{[U'|V]}(u'|v) = F_{[U']}(u')$ so, since $f_v$ is strictly
increasing, $\prag{F}_{[U'|V]}(f_v(u)|v) = F_{[U']}(f_v(u))$ and,
because the pivot is simple so $f_v$ is a bijection, $\prag{F}_{[U|V]}(u|v) =
F_{[U']}(f_v(u))$ for all $u \in \range{U \mid V=v}$, so
$\prag{F}_{[U|V]}(u|v)$ is of form (\ref{eq:belangrijker}).  

For the converse, assume again that $f_v$ is increasing and take 
$\prag{F}_{[U|V]}(u|v)$ of form (\ref{eq:belangrijker}). Then, following the steps above in backward direction, we find that all steps remain valid and show that  
for all $v \in \range{V}$, $u' \in \range{U'}$,
$\prag{F}_{[U'|V]}(u'|v) = F_{[U']}(u')$, which shows that $\tilde{P}$ is pivotally safe for $U |V$ with pivot $U'$.

{\em (1) $\Rightarrow$ (4)}.  To show that the SDA (scalar density
assumption) is satisfied note that, because $U'$ is a continuous pivot, $P(U')$
satisfies the SDA by definition; because pivotal safety holds, so does
$\tilde{P}(U'\mid V=v)$ for each $v \in \range{V}$. Because the pivot
$U$ is simple, the function $f_v$ in Definition~\ref{def:pivot} is a
bijection and it follows that $\tilde{P}('\mid V=v)$ also satisfies
SDA for each $v \in \range{V}$.

Now assume that $U'$ is an increasing pivot, i.e. the function $f_v(u) := f(u,v)$ with $U'= f(U,V)$ is increasing in $u$, for all $v \in \cV$ (the decreasing pivot case is proved analogously). 
For each $b \in [0,1]$, we have: 
\begin{align}\label{eq:mla}
& \{ z\in \cZ: \cdf(U(z) \mid V(z)) \leq b \}) = \nonumber
\{ z\in \cZ: \prag{F}_{[U']}(f(\; U(z),V(z)\; ) \mid V(z)) \leq b \} = \\
& 
\{ z\in \cZ: \prag{F}_{[U']}(f(U(z),V(z))) \leq b \} = 
\{ z\in \cZ: \prag{F}_{[U']}(U'(z)) \leq b \} = \nonumber \\
& 
\{ z\in \cZ: {F}_{[U']}(U'(z)) \leq b \},
\end{align}
where the first equality follows because $f_v$ must be strictly
increasing, the second because $U'$ is a pivot, the third is rewriting
and the fourth again because $U'$ is a pivot. Because $U'$ is a
continuous pivot, it satisfies SDA and thus, for all $P \in \cP^*$,
$F(U')$, the CDF under $P$ of $U'$, is uniform, so $P(F_{[U']}(U')
\leq b) = b$ for all $b \in [0,1]$. Using (\ref{eq:mla}) now
gives that $P(\prag{F}(U|V) \leq b) = b$.

Since as already established, $\tilde{P}(U \mid V=v)$ satisfies the
SDA, we also have $\tilde{P}( \prag{F}(U|V) \leq b \mid V=v) = b$, for
all $v \in \range{V}$.  Together these results imply that $\tilde{P}$
is safe for $\prag{F}(U|V) \mid V$.

{\em (4) $\Rightarrow$ (5)}. 
The third requirement of Definition~\ref{def:pivot} holds by assumption.  To show that the first and second requirements hold, note that by
the SDA  $f_v(u) :=
\prag{F}_{[U|V]}(u|v)$ must be continuous strictly increasing as a
function of $u$ on $\range{U}$ for all $v \in \cV$, so that
$\prag{F}(U|V)$ is a pivot, and again by the SDA, $f_v$ ranges from
$0$ to $1$ and hence it has an inverse, hence it is a bijection, so
that $\prag{F}(U|V)$ is even a simple pivot.

{\em (5) $\Rightarrow$ (1)\/} is trivial.  

{\em (3) $\Leftrightarrow$ (4)\/} Let $U'= \tilde{F}(U|V)$. By
Proposition~\ref{prop:newstart}, (\ref{eq:calibrationagain}), safety for $U' \mid [V]$ is
equivalent to having, for all $P \in \cP^*$, all $v \in \support_P(V)$,
 $P(U') = \prag{P}(U' \mid V=v)$. This in turn is equivalent to having, for all $b \in [0,1]$, 
$P(U' \leq b) = \prag{P}(U' \leq b \mid V=v)$. Since, by the SDA, $\prag{P}(U' \leq b \mid V=v) = b$, we get that safety for $U'\mid [V]$ is equivalent to having 
for all $P \in \cP^*$, all $v \in \support_P(V)$, all $b \in [0,1]$,
 $P(U' \leq b) = b$.  But this is just the same as confidence--safety for $U \mid V$ with $a=0$, which shows $(1) \Rightarrow (2)$. For the converse, we note that we have just shown that safety for $U'\mid [V]$ implies that for all $P \in \cP^*$, all $v \in \support(V)$, all $0 \leq a < b \leq 1$, that  
$P(U' \leq b) = b$ and $P(U'\leq a) = a$ whence $P(a < U' \leq b) = b-a$, implying confidence--safety.

\paragraph{Proof of Theorem~\ref{thm:discreteconfidence}}
Let $U'= \prag{p}(U|V)$ and consider the function $f$ with $U' = f(U,V)$ and let $f_v(u)$ be as in Definition~\ref{def:pivot}. The following fact is immediate by the condition of `uniqueness of nonzero probabilities' imposed on $\prag{P}(U \mid V=v)$:

{\em Fact 1}. \ 
For
each $v \in \range{V}$, $f_v$ is an injection.\\ \ 
\\ 
We then find, by Definition~\ref{def:pivot} and~\ref{def:pivotal} that $U'$ is a simple pivot and $\tilde{P}$ is pivotally safe iff for all $P \in \cP^*$, for all $v \in \range{V}$, 
$$
P(U') = \prag{P}(U') = \prag{P}(U' \mid V= v),
$$
which, by Proposition~\ref{prop:newstart}, (\ref{eq:calibrationagain}), is equivalent to
$\prag{P}$ being safe for $U' \mid [V]$. This establishes $(1)
\Leftrightarrow (2)$.
The implication $(2) \Rightarrow (3)$ is trivial (take $U'$ as
pivot). Thus it only remains to show:

{\em (3) $\Rightarrow$ (1):\/} Let $U''$ be a simple pivot for $U \mid
V$ and suppose that pivotal safety holds with pivot $U''$. We first show that for each $p \in [0,1]$, $\tpmf{U|V}(U|V) = p \Leftrightarrow \tpmf{U''}(U'') = p$, i.e.
\begin{equation}\label{eq:properfiducial}
\{z \in \cZ: \tpmf{U|V}(U(z)|V(z)) = p \} = \{z \in \cZ: \tpmf{U''}(U''(z)) = p \}. 
\end{equation}
To see this, note that, because
for each $v \in \range{V}$, the mapping
$f_{v}(u) := f(u,v)$ is a bijection from $\range{U \mid
  V=v}$ to $\range{U''}$, we have
\begin{align*}
& \{z \in \cZ: \tpmf{U|V}(U(z)|V(z)) = p \} = 
\{z \in \cZ: \tpmf{U''|V}(f_{V(z)}(U(z))|V(z)) = p \} = \\ & 
\{z \in \cZ: \tpmf{U''}(f_{V(z)}(U(z))) = p \} = 
\{z \in \cZ: \tpmf{U''}(U''(z)) = p \},
\end{align*}
where the first equality follows from Fact 1 above, the second because
of pivotal safety, which imposes that $\prag{P}(U'' \mid V=v) = \prag{P}(U'')$ for all $v \in \range{V}$, and the third by definition of $f_v(u)$. Thus,
(\ref{eq:properfiducial}) follows, and it implies that the two events in
(\ref{eq:properfiducial}) must have the same probability under any
single probability measure on $\cZ$, in particular under $\prag{P}(\cdot
\mid V=v)$ for all $v \in \range{V}$ and for all $P \in \cP^*$, i.e.
\begin{align}\label{eq:properb}
\prag{P}(\{z \in \cZ: \tpmf{U|V}(U(z)|V(z)) = p \} \mid V=v) & = 
\prag{P}(\{z \in \cZ: \tpmf{U''}(U''(z)) = p \}
 \mid V=v)
  \\ \label{eq:properc}
{P}(\{z \in \cZ: \tpmf{U|V}(U(z)|V(z)) = p \} )  & = 
{P}(\{z \in \cZ: \tpmf{U''}(U''(z)) = p \}).
\end{align}
Since  $U''$ is a pivot,
$\prag{P}(U'' \mid V=v)$ is the same for all $v \in \range{V}$ and equal
to $\prag{P}(U'')$ and also to $P(U'')$, for all $P \in
\cP^*$. Combining this with (\ref{eq:properb}) we find that
$$
\prag{P}(\{z \in \cZ: \tpmf{U|V}(U(z)|V(z)) = p \} \mid V=v)  = 
{P}(\{z \in \cZ: \tpmf{U''}(U''(z)) = p \}).
$$
Rewriting this further using (\ref{eq:properc}) gives
$$
\prag{P}(\{z \in \cZ: \tpmf{U|V}(U(z)|V(z)) = p \} \mid V=v) = {P}(\{z
\in \cZ: \tpmf{U|V}(U(z)|V(z)) = p \} ),$$ i.e., setting $U'=
\tpmf{U|V}(U|V)$, we find that for all $v \in \range{V}$,
$\prag{P}(U'\mid V=v) = P(U')$; thus $\tilde{P}$ is pivotally safe for
$U|V$ with simple pivot $U'$, and {\em (1).} follows.

\paragraph{Proof of Theorem~\ref{thm:loss}}
By assumption there is some simple pivot $U'= f(U,V)$, such that for
each $v \in \range{V}$, the function $f_v$ on $\range{U \mid V= v}$
defined as $f_v(u) = f(u,v)$ is a bijection to $\range{U'}$. We now
fix some function $g: \cU' \rightarrow \range{U}$ that is 1-to-1 (an
injection, not a bijection).  Such a function must exist; we can, for
example, take $f^{-1}_{v_0}$ for arbitrary but fixed $v_0$ which
exists because $f_{v_0}$ must be a bijection by definition.  Also note
that for any bijection $f: \cU \rightarrow \cU$ and its extension to
$\cA$ as defined in the main text, we have, for every distribution $P$
on $\cU$ with mass function $p$ and $\tilde{a}_{P(U)}$ denoting the
function from $P(U)$ to a Bayes act for $P(U)$ (which we assume to exist), by symmetry of the loss:
\begin{align*}
& \sum_{u \in \cU} P(f(u)) \cdot L(f(u),f(\tilde{a}_{P(U)})) = 
\sum_{u \in \cU} P(u) \cdot L(u,\tilde{a}_{P(U)}) = \\
& \min_{a \in \cA} \sum_{u \in \cU} P(u) \cdot L(u,{a}) = 
\min_{a \in \cA} \sum_{u \in \cU} P(f(u)) \cdot L(f(u),{a}) = 
\sum_{u \in \cU} P(f(u)) \cdot L(f(u),f(\tilde{a}_{P(f(U))})),
\end{align*}
hence, combining the leftmost and righmost expression,
\begin{equation}\label{eq:lancaster}
f(\tilde{a}_{P(U)}) = \tilde{a}_{P(f(U))}.
\end{equation}
Now repeatedly using symmetry of the loss function and (\ref{eq:lancaster}), we have:
\begin{align*}
& \sum_{u \in \cU} \prag{P}(U = u \mid V=v) \cdot L(u,\prag{a}_v) = 
\sum_{u \in \cU} \prag{P}(\{z: U(z) = u \} \mid V=v) \cdot L(u,\prag{a}_v) =
\\ & 
\sum_{u \in \cU} \prag{P}(\{z: f_v(U(z)) = f_v(u) \} \mid V=v) \cdot L(g(f_v(u)),
g(f_v(\prag{a}_v))) = \\ &
\sum_{u' \in \cU'} \prag{P}(\{z: U'(z)) = u' \} \mid V=v) \cdot L(g(u'),a_{\prag{P}(g(f_v(U)) \mid V=v)}) = \\&
\sum_{u' \in \cU'} \prag{P}(\{z: U'(z)) = u' \}) \cdot L(g(u'),a_{\prag{P}(g(U') \mid V=v)})=
\sum_{u' \in \cU'} {P}(\{z: U'(z)) = u' \}) \cdot L(g(u'),a_{\prag{P}(g(U'))})= \\ &
\sum_{u' \in \cU', v \in \range{V}} {P}(\{z: U'(z)) = u' , V(z) = v\}) \cdot L(g(u'),a_{\prag{P}(g(U'))})=  \\ & 
\sum_{u' \in \cU', v \in \range{V}} {P}(\{z: U'(z)) = u' , V(z) = v\}) \cdot L(g(u'),a_{\prag{P}(g(U')\mid V=v)})= \\ & 
\sum_{u' \in \cU', v \in \range{V}} {P}(\{z: f_v^{-1} (U'(z)) = f_v^{-1}(u') , V(z) = v\})  \cdot L(g(f_v(f_v^{-1}(u'))),a_{\prag{P}(g(f_v(U)) \mid V=v)})=
\end{align*}
\begin{align*}
& \sum_{u \in \cU, v \in \range{V}} {P}(\{z: U(z) = u , V(z) = v\}) \cdot L(g(f_v(u)),a_{\prag{P}(g(f_v(U)) \mid V=v)})= \\
& \sum_{u \in \cU, v \in \range{V}} {P}(\{z: U(z) = u , V(z) = v\}) \cdot L(g(f_v(u)),g(f_v(
a_{\prag{P}(U \mid V=v)})))= \\
& \sum_{u \in \cU, v \in \range{V}} {P}(\{z: U(z) = u , V(z) = v\})\cdot L(u,\prag{a}_v) = \Exp_{U,V \sim P} [L(U,\prag{a}_V)],
\end{align*}
and the result follows.
\end{document}